\documentclass[reprint, amsmath,amssymb,aps,longbibliography]{revtex4-1}

\usepackage{newtxtext}

\usepackage{graphicx}%
\usepackage{dcolumn}%

\usepackage[T1]{fontenc}
\usepackage[latin9]{inputenc}
\usepackage{mathtools}
\usepackage{amsmath}
\usepackage{amsthm}
\usepackage{amssymb}
\usepackage{stmaryrd}

\usepackage{chngcntr}
\usepackage{blkarray,bigstrut}

\usepackage{color}%
\RequirePackage[dvipsnames,usenames]{xcolor}
\definecolor{ForestGreen}{rgb}{0.1333,0.5451,0.1333}
\definecolor{DarkRed}{rgb}{0.8,0,0}
\definecolor{Red}{rgb}{1,0,0}
\usepackage[linktocpage=true,
pagebackref=false,colorlinks,
linkcolor=DarkRed,citecolor=ForestGreen,
bookmarks,bookmarksopen,bookmarksnumbered]
{hyperref}

\usepackage{thm-restate}
\usepackage{paralist}

\usepackage{environ}

\newif\ifappendix
\NewEnviron{maybeappendix}[1]
{\ifappendix
	\expandafter\global\expandafter\let\csname putmaybeappendix#1\endcsname\BODY%
	\else
	\expandafter\newcommand\csname putmaybeappendix#1\endcsname{}\BODY%
	\fi}
\newcommand{\putmaybeappendix}[1]{\csname putmaybeappendix#1\endcsname}
\appendixtrue

\usepackage[shortlabels]{enumitem}
\usepackage[capitalize]{cleveref}

\makeatletter

\DeclareMathOperator*{\argmin}{arg\,min}
\DeclareMathOperator*{\supp}{supp}

\newcommand{\eq}[1]{\begin{align}#1\end{align}}
\newcommand{\cs}[1]{\mathsf{#1}}   %

\renewcommand{\mid}{\vert}

\renewcommand{\r}{r}
\newcommand{\rr}{r}
\newcommand{\map}{P}

\newcommand{\bathix}{\alpha}

\newcommand{\Cin}{{\mathrm{IN}}}
\newcommand{\Cout}{{\mathrm{OUT}}}
\newcommand{\MM}{{\mathcal{M}}}
\newcommand{\RR}{{\mathcal{R}}}

\newcommand{\N}{{\mathbb{N}}}
\newcommand{\Q}{{\mathcal{Q}}}
\newcommand{\totrans}{\! \shortrightarrow \!}

\newcommand{\LL}{\mathcal{L}}
\newcommand{\LLL}{L}
\global\long\def\CC{\Phi}%
\global\long\def\l{c}%
\global\long\def\FF{\Gamma}%

\global\long\def\Lmap{\LLL(\map)}%

\newcommand{\QCC}{\Q_\mathrm{circ}}
\newcommand{\Qao}{\Q_\mathrm{AO}}
\newcommand{\DMcirc}{\MM^\mathrm{loss}}

\newcommand{\LLC}{{{\mathcal{L}}^\mathrm{circ}}}
\newcommand{\LandauerLoss}{\mathcal{L}^\mathrm{loss}}

\newcommand{\EP}{\sigma}
\newcommand{\subEP}{\hat{\sigma}}

\newcommand{\nullstate}{{\emptyset}}

\newcommand{\EPg}{{\subEP_g}}
\newcommand{\qPP}{\q}
\newcommand{\qcPP}{\q^{c}}

\newcommand{\mC}{\pi_\CC}
\newcommand{\mG}{\pi_g}

\newcommand{\DDbase}{D}
\newcommand{\IIbase}{\mathcal{I}}

\newcommand{\SSS}{S}

\newcommand{\DDf}[2]{\DDbase(#1\Vert #2)}
\newcommand{\DDb}[2]{\DDbase\big(#1\Vert #2\big)}
\newcommand{\II}{\IIbase}
\newcommand{\IIDD}{\mathcal{D}}

\newcommand{\IIDDf}[2]{\mathcal{D}(#1\Vert #2)}

\newcommand{\pa}{\mathrm{pa}}
\newcommand{\pag}{{\pa(g)}}
\newcommand{\begG}{{\mathrm{beg}(g)}}
\newcommand{\ennG}{{\mathrm{end}(g)}}

\newcommand{\q}{{q}}
\newcommand{\ppag}{p_\pag}
\newcommand{\qpag}{\q_\pag}
\newcommand{\SG}{{n(g)}}
\newcommand{\pSG}{p_\SG}
\newcommand{\qSG}{\q_\SG}

\newcommand{\pIN}{p_\Cin}
\newcommand{\ppIN}{q_\Cin}
\newcommand{\ppPAG}{q_\pag}

\newcommand{\sX}{\mathcal{X}}
\newcommand{\sY}{\mathcal{Y}}
\newcommand{\sZ}{\mathcal{Z}}

\newcommand{\pA}{p_A}
\newcommand{\PA}{P_A}

\newcommand{\xA}{x_A}
\newcommand{\xB}{x_B}
\newcommand{\sXA}{\sX_A}
\newcommand{\sXB}{\sX_B}
\newcommand{\sXAB}{\sXA \times \sXB}

\newcommand{\W}{K}

\newcommand{\variablevalue}{{x}}

\makeatother

\newtheorem{apptheorem}{Theorem}

\newtheorem{theorem}{Theorem}
\newtheorem{corollary}[theorem]{Corollary}
\newtheorem{example}{Example}
\newtheorem{applemma}{Lemma}

\newcommand{\apprefcost}{Appendix A} %
\newcommand{\apprefpartialsupport}{Appendix A} %
\newcommand{\apprefsinglepass}{Appendix B}
\newcommand{\apprefoutdegreebig}{Appendix C} %

\begin{document}

	\title{Thermodynamics of computing with circuits}

	\date{\today}%

	\author{David H. Wolpert}
		\affiliation{Santa Fe Institute, Santa Fe, New Mexico}
	\altaffiliation[Also at ]{Complexity Science Hub, Vienna} 
	\altaffiliation[ ]{Arizona State University, Tempe, Arizona.}

	\author{Artemy Kolchinsky}
	
	\affiliation{Santa Fe Institute, Santa Fe, New Mexico}
	
	\begin{abstract}
Digital computers implement computations using circuits, as
do many naturally occurring systems (e.g., gene regulatory networks).
The topology of any such circuit restricts which variables 
may be physically coupled during the operation of a circuit. We investigate how  such
restrictions on the physical coupling affects the thermodynamic costs of running the circuit. 
 To do this we first calculate the minimal additional entropy production that 
arises when we run a given gate in a circuit.
We then build on this calculation, to analyze how the thermodynamic costs of implementing a 
computation with a full circuit, comprising multiple connected gates, depends on the topology of that circuit. 
This analysis provides a rich new set of optimization
problems that must be addressed by any designer of a circuit, if they wish to 
minimize thermodynamic costs. 
	\end{abstract}

	\maketitle

	\section{Introduction}

		A long-standing focus of research in the physics community
		has been how the \textit{energetic} resources required to perform a given computation depend on that computation.
		This issue is sometimes referred to as the
	``thermodynamics of computation'' or the ``physics of information''~\cite{benn82,parrondo2015thermodynamics,wolpert_thermo_comp_review_2019}.
%
		Similarly,  a central focus of computer science
	theory has been how the minimal \textit{computational} resources needed to perform a given computation
	depend on that computation~\cite{arora2009computational,savage1998models}. 
	(Indeed, some of the most important open issues in computer science, like whether $\cs{P = NP}$, 
	concern the relationship between a computation and its resource requirements.)
	Reflecting this commonality of interests, there was a burst of early research relating the 
	resource concerns of computer science theory with the resource concerns of
	thermodynamics~\cite{fredkin1982conservative,lloyd1989use,caves1993information,lloyd2000ultimate}\footnote{Like this early work, in this paper we focus on
		computation by classical systems; see~\cite{Gour2015,Brandao2013} and associated articles for recent work on the resources required to perform quantum
		computations.}.

	Starting a few decades after this early research, there was dramatic progress in our understanding of
	non-equilibrium statistical physics~\cite{jarzynski1997nonequilibrium,crooks1998nonequilibrium,hasegawa2010generalization,takara_generalization_2010,seifert2012stochastic,parrondo2015thermodynamics}, which
has resulted in new insights into the thermodynamics of  
computation~\cite{sagawa2014thermodynamic,parrondo2015thermodynamics,hasegawa2010generalization,wolpert_thermo_comp_review_2019}. 
In particular, recent research has derived the ``(generalized) Landauer
bound''~\cite{maroney2009generalizing,turgut_relations_2009,faist2015minimal,wolpert_arxiv_beyond_bit_erasure_2015,owen_number_2018,wolpert_spacetime_2019}, which states that the heat generated by a
thermodynamically reversible process that sends an initial distribution $p_0(x_0)$ to an ending distribution $p_1(x_1)$
	is $k T [S(p_0) - S(p_1)]$ (where $S(p)$ indicates the entropy of distribution $p$, $T$ is the temperature of the single bath,
and $k$ is Boltzmann's constant).

	Almost all of this  work on the Landauer bound
	assumes that the map taking initial states to final states, $P(x_1 | x_0)$, is implemented with 
	a monolithic, ``all-at-once'' physical process, jointly evolving all of the variables in the system at once.
	In contrast, for purely practical reasons
	modern %
	computers are built out of %
	{circuits}, i.e.,  they are built out of networks of ``gates'', each of which 
evolves only a small subset of  the variables of the full system~\cite{arora2009computational,savage1998models}.  
An example of a simple circuit that computes the parity of 3 input bits using two XOR gates, and which we will return to throughout this paper, is illustrated in \cref{fig:1}.

	Similarly, in the natural world, biological cellular regulatory networks carry out complicated computations by decomposing them into circuits
of far simpler computations~\cite{wang2012boolean,yokobayashi2002directed,brophy2014principles}, as do many other kinds of biological 
	systems~\cite{qian2011scaling,clune2013evolutionary,melo2016modularity,deem2013statistical}.

As elaborated below, there are two major, unavoidable thermodynamic effects of implementing a given computation 
with a circuit of gates rather than with an all-at-once process:

\vspace{5pt}

\noindent I) 
Suppose we build a circuit out of
gates which were manufactured without any specific circuit in mind.
Consider such a gate that implements bit erasure, and suppose that it
is thermodynamically reversible if $p_0$
is uniform. So by the Landauer bound,
it will generate heat $kT S(p_0) = kT \ln 2$\textit{ if run on a uniform distribution}. 

Now in general, depending on where such a bit-erasing gate appears in a 
circuit, the actual initial distribution of its states, $p'_0$, will be non-uniform.
This not only changes the Landauer bound for that gate from $kT \ln 2$ to $kT S( p'_0)$; it is now known that since the gate is thermodynamically
reversible for $p_0 \ne p'_0$, running that gate on $p'_0$ will \textit{not} be thermodynamically
reversible~\cite{kolchinsky2016dependence}. So the actual heat generated by running that bit will exceed the associated value
of the Landauer bound, $kT S(p'_0)$.

%
%
%
%
%
%
%
%

\vspace{5pt}

\noindent II) Suppose the circuit 
is built out of two bit-erasing gates, and that each gate is thermodynamically reversible on a uniform input distribution when run separately from the circuit.
If the {marginal} distributions over the initial states of the gates are both uniform, then the heat generated by running each
of them is $kT \ln 2$, and therefore the total generated heat is $2 kT \ln 2$. Suppose though that there is nonzero statistical coupling
between their states under their initial {joint} distribution. Then as elaborated below, even though each of the gates run
separately is thermodynamically reversible, running them in parallel is \textit{not} thermodynamically reversible. So running them
generates extra heat beyond the minimum given by applying the Landauer bound to the dynamics of the full joint distribution~\footnote{For
example, if the initial states of the gates are perfectly correlated, the initial entropy of the two-gate system is $\ln 2$. In this case
the running the gates in parallel rather than in a joint system generates extra heat of $2 kT\ln 2 - kT\ln 2$, above the minimum possible given by the Landauer bound.}.

\vspace{5pt}

\begin{figure}
  \includegraphics[width=0.5\linewidth]{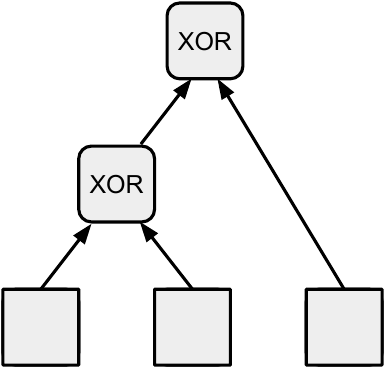}
  \caption{A simple circuit that uses two exclusive-OR (XOR) gates to compute the parity of 3 inputs bits. The circuit outputs a 1 if an odd number of input bits are set to 1, and a 0 otherwise.}
  \label{fig:1}
\end{figure}

	These two effects mean that
	the thermodynamic cost of running %
	a given computation with a circuit will in general vary greatly depending on the precise circuit we 
	use to implement that computation. %
	In the current paper we 
	%
	%
	analyze this dependence.  
	
	We make no restriction on the input-output maps computed by each gate
	in the circuit. They can be either deterministic (i.e., single-valued) or 
	stochastic, %
	logically reversible (i.e., implementing a deterministic permutation of
the system's state space, as in Fredkin gates~\cite{fredkin1982conservative}) or not, etc. 
	However, to ground thinking, the reader may imagine  
	that the circuit being considered is a Boolean circuit, where 
	each gate performs one of the usual single-valued Boolean functions, like logical AND gates, XOR gates, etc. 
	
	For simplicity, in this paper we focus on %
	%
	%
circuits whose topology does not contain loops~\cite{savage1998models,wegener_complexity_1991}, such as the circuit shown in \cref{fig:1}.

	\subsection{Contributions}

	We have four primary contributions.

	\noindent \textbf{1)} We derive exact expressions for how
	the entropy flow (EF) and entropy production (EP) produced by a fixed dynamical system %
	vary as one changes the initial distribution of states of that system.  
These expressions capture effect (I) described above.
(These expressions extend an earlier analysis~\cite{kolchinsky2016dependence}).
%
%
%
%
%
%
%
%
%
%
%
%
%
%
%

\noindent \textbf{2)}
We introduce ``solitary processes''. These are a type of physical process that can implement
any particular gate in a circuit while respecting the constraints on what variables 
in the rest of the circuit that
gate is coupled with.  We can use the thermodynamic properties of solitary processes to analyze effect (II) described above.
	%
	
	%
\noindent \textbf{3)}
We combine our first two contributions
	to analyze 
the thermodynamic costs of implementing circuits in a
	 ``serial-reinitializing'' manner. This means two things: the gates in the circuit are 
	 run one at a time, so each gate is run as a solitary process; after a gate is run its input
	 wires are reinitialized, allowing for subsequent reuse of the circuit.  
	In particular, we derive expressions relating the minimal EP generated by running an SR circuit to
information-theoretic quantities associated with the wiring diagram of the circuit.
	%
	%
	
%
	%
	%
	%

\noindent \textbf{4)}
Our last contribution is an expression for the extra EP 
	that arises in running an SR circuit if the initial state distributions at its gates differ from the ones that
	result in minimal EP for each of those gates. This expression
involves an information-theoretic function that we call ``multi-divergence'' which appears to be new to the
literature.

\subsection{Roadmap}

In \cref{sec:notation} we introduce general notation, and then provide a minimal summary of the parts of stochastic thermodynamics,
information theory and circuit theory that will be used in this paper. We also introduce the definition
of the ``islands'' of a stochastic matrix in that section, which will play a central role in our analysis. In 
\cref{sec:efdecomp} we derive an exact expression for how the EF and EP
of an arbitrary process depends on its initial state distribution.
In \cref{sec:solitary} we introduce solitary processes
and then analyze their thermodynamics.
%
%
In \cref{sec:singlepass} we introduce SR  circuits.
In \cref{sec:singlepasscosts} we use the tools
 developed in the previous sections to analyze the
thermodynamic properties of SR circuits. %
In \cref{sec:earlierwork} %
we discuss related earlier work. 
\cref{sec:future_work} concludes and presents some directions for future work.  
All proofs %
that
are longer than several lines are collected in the appendices.

\section{Background}

Because the analysis of the thermodynamics of circuits involves tools from multiple fields,
we review those tools in this section. We also introduce
some new mathematical structures that will be central to our analysis, in particular ``islands''.
We begin by introducing notation.

\subsection{General notation}
\label{sec:notation}

We write a Kronecker delta as $\delta(a, b)$. 
We write a random variable with an upper
case letter (e.g., $X$), and the associated set of possible outcomes with the associated calligraphic letter (e.g., $\sX$). A particular outcome of a random
variable is written with a lower case letter (e.g., $x$).  We also use lower case letters
like $p$, $q$, etc. to indicate probability distributions.  

We  use $\Delta_\sX$ to indicate the set of probability distribution over a set of outcomes $\sX$. For any distribution $p \in \Delta_\sX$, we use $\supp p:=\{ x\in \sX : p(x) > 0\}$ to indicate the support of $p$. 
Given a distribution $p$ over $\sX$ and any $\sZ \subseteq \sX$, we write $p(\sZ) = \sum_{x\in \sZ} p(x)$ to indicate the probability that the outcome of $X$ is in $\sZ$. Given a function $f : \sX \to \mathbb{R}$, we write $\mathbb{E}_p[f]$ to indicate $\sum_x p(x) f(x)$, the expectation of $f$ under distribution $p$.

Given any conditional distribution $\map(y \vert x)$ of $y \in \sY$ given  $x \in \sX$, %
and some distribution $p$ over $\sX$, we write $\map p$ for the distribution over $\sY$
induced by applying $\map$ to $p$:
\begin{align}
[\map p](y) := \sum_{x\in \sX} \map(y\vert x) p(x) \,.
\label{eq:matrixnotation}
\end{align}
We will sometimes use the term ``map'' to refer to a conditional distribution. 

We say that a conditional distribution $\map$ is ``logically reversible'' if it is deterministic (the entries of $\map(y\vert x) $ are 0/1-valued for all $x\in\sX$ and $y\in \sY$) and if there do not exist $x,x'\in\sX$ and $y\in \sY$ such that $\map(y\vert x)>0$ and $\map(y\vert x')>0$. When $\sY=\sX$, a logically reversible $\map$ is simply a permutation matrix. Given any subset of states $\sZ\subseteq \sX$, we also say that $\map$ is ``logically reversible over $\sZ$'' if the entries $\map(y\vert x)$ are 0/1-valued for all $x\in \sZ$ and $y\in\sY$, and there do not exist $x,x'\in\sZ$ and $y\in \sY$ such that $\map(y\vert x)>0$ and $\map(y\vert x')>0$.

We write a multivariate random variable with components $V=\{1,2,\dots\}$ as $X_V = (X_1, X_2, \dots)$, with outcomes $x_V$. 
We will also use upper case letters (e.g., $A,V,\dots$) to indicate sets of variables. 
For any
subset $A \subseteq V$ we use the random variable $X_A$ (and its outcomes $x_A$) to
refer to the components of $X_V$ indexed by $A$. Similarly, for a distribution $p_V$ over
$X_V$, we write the marginal distribution over $X_A$ as $p_A$. 
For a singleton set $\{a\}$, we slightly abuse notation and write $X_a$ instead of $X_{\{a\}}$.

\subsection{Stochastic thermodynamics}
\label{sec:stoch_thermo}

We will consider a circuit to be 
physical system
in contact with one or more thermodynamic reservoirs (heat baths, chemical baths, etc.). 
The system evolves over some time interval (sometimes implicitly taken to be $t\in [0,1]$, where the units of time are arbitrary), 
possibly while being driven by a work reservoir. 
We refer to the set of thermodynamic reservoirs and the driving --- and, in particular, the stochastic dynamics they induce over the system during $t \in [0, 1]$ --- as a \textbf{physical process}.

We use $\sX$ to indicate the finite state space of the system. 
%
Physically, the states $x \in \sX$ can either be microstates 
or they can be coarse-grained macrostates
under some additional assumptions (e.g., that all macrostates have the same ``internal entropy''~\cite{wolpert2016free,wolpert_arxiv_beyond_bit_erasure_2015,parrondo2015thermodynamics}). 
While we
are ultimately interested in the special case where  the system is a circuit with a set of nodes $V$, the review
in this section is more general.

\def\Wt{\W_t}

While much of our analysis applies more broadly, %
to make things concrete one may imagine that the system undergoes master equation dynamics,
also known as a continuous-time Markov chain (CTMC), as often used in stochastic
thermodynamics to model discrete-state physical systems.  In this
subsection we briefly review stochastic thermodynamics, 
referring the reader to \cite{van2015ensemble,esposito2010three} for more details.

Under a CTMC, 
the probability distribution over $\sX$ at time $t$, indicated by $p^t$, evolves according to the  %
master equation
\begin{align}
\frac{d }{dt} p^t(x') = \sum_{{x}}p^t(x) \Wt(x \totrans x') ,
\label{eq:ctmc_dynamics2}
\end{align}
where $\Wt$ is the rate matrix at time $t$. For any rate matrix $\Wt$, the off-diagonal entries $\Wt(x \totrans x')$ (for $x \ne x'$)  indicate  the rate at which probability flows from state $x$ to $x'$, while the diagonal entries are fixed by $\Wt(x \totrans x)=-\sum_{x' (\ne x)} K_t(x \totrans x')$, which guarantees conservation of probability. %
If the system is connected to multiple thermodynamic reservoirs indexed by $\bathix$, the rate matrix can be further decomposed as $\Wt(x \totrans x') = \sum_\bathix \Wt^\bathix(x \totrans x')$, where $\Wt^\bathix$ is the rate matrix at time $t$ corresponding to reservoir $\bathix$.  

The term \textbf{entropy flow} (EF) refers to the increase of entropy in all coupled reservoirs. 
The instantaneous rate of EF out of the system 
at time $t$ is defined as
\begin{align}
\dot{\Q}(p^t) =\sum_{\bathix, x,x'} p^t(x) \Wt^\bathix(x \totrans x') \ln \frac{\Wt^\bathix(x \totrans x')}{\Wt^\bathix(x' \totrans x)} .
\label{eq:defefrate0}
\end{align}
The overall EF incurred over the course of the entire process is $\Q = \int_0^1 \dot{\Q} \; dt$. 

The term \textbf{entropy production} (EP) refers to the overall increase of entropy, both in the system and in all coupled reservoirs. The instantaneous rate of EP 
at time $t$ is defined as
\begin{align}
\dot{\EP}(p^t) = \textstyle{\frac{d}{dt}} S(p^t) + \dot{\Q}(p^t) \,.\label{eq:defeprate}
\end{align}
The overall EP incurred over the course of the entire process is $\EP = \int_0^1 \dot{\EP} \;  dt$.

Note that we use terms like ``EF'' and ``EP'' to refer to either the associated rate or the associated integral over a non-infinitesimal time interval; the context should always make
the precise meaning clear.

Given some initial distribution $p$, the EF, EP, and the drop in the entropy of the system from the beginning to the end of the process are related according to
\begin{align}
\Q(p) = \left[  \SSS(p) - \SSS(\map p) \right] + \EP(p) \,.
\label{eq:newdecomp1}
\end{align}
In general, the EF can be written as the expectation $\Q(p)=\sum_x p(x) q(x)$, where $q(x)$ indicates the expected EF arising from trajectories that begin on state $x$. 
Given that the drop in entropy is a nonlinear function of $p$, while the expectation $\Q(p)$ is a linear function of $p$, 
\cref{eq:newdecomp1} tells us that EP will generally be a nonlinear function of 
$p$. Note that if $P$ is logically reversible, 
then $S(p)=S(\map p)$ and therefore EF and EP will be equal for any $p$.

While the EF can be positive or negative, 
the log-sum inequality can be used to prove that EP for master equation dynamics is 
non-negative~\cite{esposito2011second,seifert2012stochastic}:
\begin{align}
\Q(p) \ge \SSS(p) - \SSS(\map p) \,.
\label{eq:secondlaw}
\end{align}
This can be viewed as a derivation of the second law of thermodynamics, given the assumption that our
system is evolving forward in time as a CTMC.

All of these results are purely mathematical and hold for any CTMC dynamics, even in contexts having nothing to do with
physical systems. However, these results can be interpreted in 
thermodynamic terms when each $\Wt^\bathix$  obeys \textbf{local detailed balance} 
(LDB) with regard to %
thermodynamic
reservoir $\bathix$~\cite{van2015ensemble,seifert2012stochastic,wolpert_thermo_comp_review_2019}. 
Consider a system with Hamiltonian $H_t(\cdot)$ at time $t$, and let $\bathix$ label a heat
bath whose inverse temperature is $\beta_\bathix$.  
Then, $\Wt^\bathix$ will obey LDB when for all $x,x'\in \sX$, either $\W_t^\bathix(x  \totrans  x')= \W_t^\bathix(x' \totrans  x)= 0$, or 
\begin{align}
\label{eq:ldb}
\frac{\Wt^\bathix(x  \totrans  x')}{\Wt^\bathix(x'  \totrans  x)} = e^{\beta_\bathix(H_t(x) - H_t(x'))}.
\end{align}
If LDB holds, then  
EF can be written as~\cite{esposito2010three}
\begin{align}
\Q(p) = \sum_\bathix \beta_\bathix Q_\bathix(p) \,,
\label{eq:efheat}
\end{align}
where 
$Q_\bathix$ is the expected amount of heat transfered from the system into bath $\bathix$ 
during the process. %

We end with two caveats concerning the use of stochastic thermodynamics to analyze real-world circuits. 
First, many of the processes described in this paper require that some transition rates
be exactly zero at some moments. In many physical models this implies
there are infinite energy barriers at those times. In addition, perfectly carrying out any deterministic map (such as bit erasure) requires the use of infinite energy gaps between some states at some times.  
Thus, as is conventional (though implicit) in much of the thermodynamics of computation literature, 
the
thermodynamic costs derived in this paper should be understood as limiting values.

Second, there are some conditional distributions that take
the system state at time $0$ to its state at time $1$, $\map(x_{1} \vert x_0)$, that cannot be implemented by any CTMC~\cite{lencastre2016empirical,kingman_imbedding_1962}. For example,  
one cannot carry out (or even approximate) a simple bit flip $\map(x_{1} \vert  x_0)=1-\delta(x_1,x_0)$ with a CTMC.
Now, we \textit{can} design a CTMC to implement
any given $P(x_{1} | x_0)$ to arbitrary precision, if the dynamics is expanded
to include a set of ``hidden states'' in addition to the states in $X$~\cite{owen_number_2018,wolpert_spacetime_2019}. 
However, as we
explicitly demonstrate below,  
SR circuits can be implemented without introducing any such hidden states; this is one of their advantages. (See also \cref{ex:ctmchiddencost} in {\apprefpartialsupport}.)

\subsection{Information theory}
\label{sec:info_notation}

Given two distributions $p$ and $r$ over random variable $X$, we use notation like
$\SSS(p)$ for Shannon entropy and $\DDf{p}{\r}$ for Kullback-Leibler (KL) divergence. 
We write $S(\map p)$ to refer to the entropy of the distribution over $Y$ induced by $p(x)$ and the conditional distribution $\map$, as defined in \cref{eq:matrixnotation}, and similarly for other information-theoretic measures.
Given two random variables $X$ and $Y$ with joint distribution $p$, we write 
{$ \SSS(p(X \mid Y))$ for the conditional entropy of $X$ given $Y$, and}
$I_p(X;Y)$ for the mutual information (we drop the subscript $p$ where the distribution is clear from context). All information-theoretic measures are in nats.

Some of our results below are formulated in terms of
an extension of mutual information to more than two random variables that
is known as ``total correlation'' or \textbf{multi-information}~\cite{watanabe1960information}. For a random variable $X_{A}=(X_1, X_2,\dots)$, the multi-information is defined as
\begin{align}
 \II(p_A) =  \Bigg[ \sum_{v \in A} \SSS(p_v) \Bigg]  -\SSS(p_A )  \,.
 \label{eq:multiinfodef}
 \end{align}

Some of the results reviewed below are formulated in terms of
the \textbf{multi-divergence} between two probability distributions over the same multi-dimensional space. This is a 
recently introduced information-theoretic measure  
which can be viewed as an extension of multi-information to include a reference distribution. 
Given two distributions $p_A$ and $r_A$ over $X_A$, the multi-divergence is defined as
\begin{align}
\IIDDf{p_A}{\r_A} := \DDf{p_A}{ \r_A} - \sum_{v \in A}  \DDf{p_v}{ \r_v} \,.
\label{eq:15a}	
\end{align}
Multi-divergence measures how much of the divergence between $p_A$ and $\r_A$ arises from the correlations among the variables $X_1,X_2,\ldots$, rather than from the marginal distributions of
each variable considered separately. See App.\,A of \cite{wolpert_thermo_comp_review_2019} for a
discussion of the elementary properties of multi-divergence and
its relation to conventional multi-information.
Note that multi-divergence is defined with ``the opposite sign'' of multi-information,
i.e., by subtracting a sum of terms involving marginal variables from a term involving
the joint random variable, rather than vice-versa.

\subsection{`Island' decomposition of a conditional distribution}
\label{sec:island_def}

A central part of our analysis will involve the equivalence relation,
\begin{align}
x \sim x' \quad \Leftrightarrow \quad \exists  y \; : \; \map(y \mid x) > 0, \map(y \mid x') > 0 .
\label{eq:islanddef}
\end{align}
In words, $x \sim x'$ if there is a non-zero probability of transitioning to some state $y$ from both $x$ and $x'$ under the conditional distribution  $\map(y \mid x)$. 
We define an \textbf{island}  of the conditional distribution $\map(y \mid x)$ as any connected subset of $\sX$ given by the transitive closure of %
this equivalence relation.
The set of islands of any $\map(\cdot \mid \cdot)$ form a partition of $\sX$,
which we write as $L(\map)$. 
%
%
%
%
%
%
%


We will also use the notion of the islands of the conditional distribution $\map$ restricted to some subset of states $\sZ \subseteq \sX$.  We write $L_\sZ(\map)$ to indicate the partition of $\sZ$ generated by the transitive closure of the relation given by \cref{eq:islanddef} for $x,x' \in \sZ$.  Note that in this notation, $L(\map) = L_\sX(\map)$.

\begin{figure}
  \includegraphics[width=0.9\linewidth]{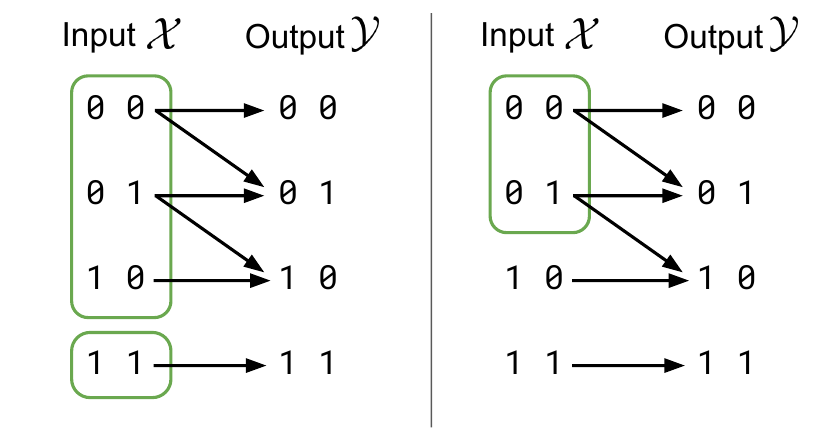}
  \caption{Left: the island decomposition for the conditional distribution in \cref{eq:mapis}. The two islands are indicated by the two rounded green boxes. Right: the island decomposition for this map with $\sX$ restricted to the subset of states $\sZ=\{\mathsf{00},\mathsf{01}\}$. For this subset of states, there is only one island, indicated by the round green box.
  } 
  \label{fig:islands}
\end{figure}

As an example, if $\map(y \vert x)>0$ for all $x \in \sX$ and $y \in \sY$ (i.e., any final state $y$ can be reached from any initial state $x$ with non-zero probability), then $L(\map)$ contains only a single island. %
As another example, if $\map(y \mid x)$ implements a deterministic
function $f : \sX \rightarrow \sY$, then $L(\map)$ is the partition of $\sX$
given by the pre-images of $f$, $L(\map) = \{ f^{-1}(y) : y \in \sY\}$. %
For example, the conditional distribution that implements the logical AND operation,%
\begin{align}
\map(c \vert a,b) = \delta(c,a b)
\end{align}
has two islands, corresponding to $(a,b) \in \{(0,0),\allowbreak (0,1),\allowbreak (1,0)\}$ and $(a,b) \in \{(1,1) \}$, respectively.  As a final example, let $\map$ be the following conditional distribution:
\begin{align}
\map(y\vert x) = \begin{bmatrix}
0.5 & 0.5 & 0  &0 \\
 0 & 0.5 & 0.5 & 0 \\
0 & 0 & 1 & 0 \\
 0 & 0 & 0  & 1 
\end{bmatrix},
\label{eq:mapis}
\end{align}
where the rows and columns corresponds to the ordered states $\sX = \sY = \{\mathsf{00},\mathsf{01},\mathsf{10},\mathsf{11}\}$.  The island decomposition for this map is illustrated in \cref{fig:islands} (left).  We also show the island decomposition for this map restricted to subset of states $\sZ=\{\mathsf{00},\mathsf{01}\}$, in \cref{fig:islands} (right).

For any distribution $p$ over $\sX$, any $\sZ\subseteq \sX$, and any $\l \in L_\sZ(\map)$, 
$p(c)= \sum_{x\in c} p(x)$ is the probability that the state of the system is contained in island $c$. It will be helpful to use the unusual notation   
$p^c(x)$ to indicate the conditional probability of $x$ within island $c$. Formally, 
$p^c(x) = p(x)/p(c)$ if $x\in c$, and $p^c(x) = 0$ otherwise.

Intuitively, the islands of a conditional distribution
are ``firewalled'' subsystems, both computationally
and thermodynamically isolated from one another for the duration
of the process implementing that conditional distribution. In particular, we will show below that 
the EP of running  $\map(y \mid x)$  on an initial distribution $p$ 
can be written as a weighted sum of the EPs involved in running $P$ on each separate 
island  $\l \in \Lmap$, where the weight for island $c$ is given by  $p(c)$.

\subsection{Circuit theory}
\label{sec:circuit_theory}

For the purposes of this paper, a \textbf{(logical) circuit} is a special type of Bayes net~\cite{kofr09,ito2013information,ito_information_2015}.
Specifically, we define any circuit $\CC$ as a tuple $(V, E, F, \sX_V)$. The pair $(V, E)$
specifies the vertices and edges of a directed acyclic graph (DAG). (We sometimes call this DAG
the \textbf{wiring diagram} of the circuit.)
$\sX_V$ is a Cartesian product $\prod_v {\sX}_v$, where each ${\sX}_v$ is the set of possible states %
associated with node $v$. $F$ 
is a set of conditional distributions, indicating the logical maps implemented at the non-root nodes of the DAG.

Following the convention in the Bayes nets literature, we orient edges in the direction of information flow. %
Thus, the inputs to the circuit are the roots of the
associated DAG and the outputs are the leaves of the 
DAG \footnote{The reader should be warned that much of the computer
science literature adopts the opposite convention.}. 
Without loss of generality, we assume that each node $v$ has a special ``initialized state'',  indicated as $\nullstate$.

We use the term \textbf{gate} to refer to any non-root node, \textbf{input node} to refer to any root node, and \textbf{output node} or \textbf{output gate} to refer to a leaf node.  
For simplicity,
we assume that all output nodes are gates, i.e., there is no root node which is also a leaf node. 
We write $\Cin$ and $x_\Cin$ to indicate the set of input nodes and their joint state, 
and similarly write $\Cout$ and $x_\Cout$ for the output nodes. %

We write the set of all gates in a given circuit as $G \subseteq V$,
and use $g \in G$ to indicate a particular gate.
We indicate the set of all nodes that are parents of gate $g$ as $\pag$. 
We indicate the set of nodes that includes gate $g$ and all parents of $g$ as $\SG := \{g\}\cup\pag$.

As mentioned, $F$ is a set of conditional distributions, indicating the logical maps implemented by each gate of the circuit.  The element of $F$ corresponding to gate $g$ is written as $\mG(x_g \vert x_\pag)$.
In conventional circuit theory, each $\mG$ is required to be deterministic (i.e., 0/1-valued). However,
we make no such restriction in this paper. 
We write the overall conditional distribution of output gates given input nodes 
implemented by the circuit $\CC$ as  
\begin{align}
\mC( x_\Cout \vert x_\Cin) = \sum_{x_{G \setminus \Cout}} \prod_{g \in G} \mG(x_g \vert x_{\pag})  \,.
\label{eq:defmc}
\end{align}

We can illustrate this formalism using the parity circuit shown in \cref{fig:1}.  Here, $V$ has 5 nodes, corresponding to the 3 input nodes and the two gates. The circuit operates over bits, so $\sX_v = \{0,1\}$ for each $v \in V$. Both gates carry out the XOR operation, so both elements of $F$ are given by $\mG(x_g \vert x_\pag) = \delta(x_g, \mathsf{XOR}(x_\pag))$ (where $\mathsf{XOR}(x_\pag)=1$ when the two parents of gate $g$ are in different states, and $\mathsf{XOR}(x_\pag)=0$ otherwise). Finally, 
 $E$ has four elements representing the edges connecting the nodes in $V$, which are shown as arrows in \cref{fig:1}.

In the conventional 
representation of a physical circuit as a (Bayes net) DAG, the wires in the physical circuit
are identified with edges in the DAG. %
However, %
in order to account for
the thermodynamic costs of communication between gates, it will be useful to represent the wires themselves as a special kind of 
gate. This means that the DAG $(V, E)$ we use to represent a  particular physical circuit
is not the same as the DAG $(V', E')$ that would be used in the conventional computer science
representation of that circuit. Rather $(V, E)$ is constructed from $(V', E')$ as follows. 

To begin, $V = V'$ and $E = E'$. Then, for each edge $(v \rightarrow \tilde{v})
\in E'$,
we first add a \textbf{wire gate} $w$ to $V$, 
and then add two edges to $E$: an edge from $v$ to $w$ and an edge from $w$ to $\tilde{v}$. 
So 
a wire gate $w$ has a single parent and a single child, and 
implements the identity map, $\pi_{w}(x_{w} \vert x_{{\rm{pa}}(w)}) = 
\delta(x_{w}, x_{{\rm{pa}}(w)})$. (This is an idealization of the real world, in which wires have nonzero probability of introducing errors.) 
We sometimes calls $(V, E)$ the \textbf{wired circuit}, to distinguish it from the original,
\textbf{logical circuit} defined as in computer science theory, $(V', E')$. We  use %
$W\subset G$ to indicate the set of wire gates in a wired
circuit. %

Every edge in a wired circuit either connects a wire gate to a non-wire gate
or vice versa. Physically, the edges of the DAG of a wired circuit
don't represent interconnects (e.g., copper wires), as they do in a logical circuit. Rather they only
indicate physical identity: an edge $e \in E$ going into a wire gate $w$ 
from a non-wire node $v$ indicates that the same physical variable will be
written as either $X_v$ or $X_{\pa(w)}$. 
Similarly, an edge $e \in E$ going into a non-wire gate $g$ 
from a wire gate $w$ indicates that $X_w$ is the same physical variable (and so always has the same state) as the corresponding
component of $X_\pag$.   
However, despite this modified meaning of the nodes in a wired circuit, \cref{eq:defmc} still applies to any wired circuit, as well
as applying to the corresponding logical circuit.
In \cref{fig:wiredcircuit}, we demonstrate how to represent the 3-bit parity circuit from \cref{fig:1} as a wired circuit.

\begin{figure}
  \includegraphics[width=0.5\linewidth]{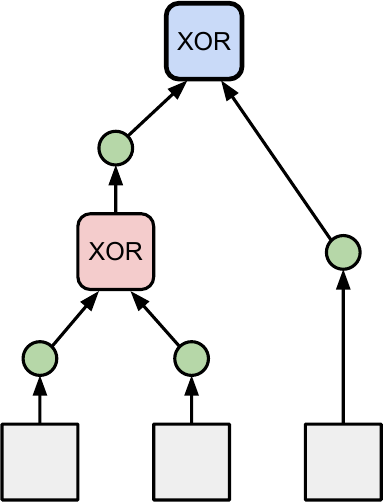}
  \caption{The 3-bit parity circuit of \cref{fig:1} represented as a wired circuit. Squares represent input nodes, rounded boxes represent non-wire gates, and smaller green circles represent wire gates. The output XOR gate is in blue, while the other (non-output) XOR gate is in red.} 
  \label{fig:wiredcircuit}
\end{figure}

We use the word ``circuit'' to refer to either an abstract wired (or logical) circuit, or to a physical system that implements that 
abstraction. Note that there are many details of the physical system that are not specified in the associated abstract circuit. 
When we need to distinguish the abstraction from its physical implementation, %
we will refer to the latter as a \textbf{physical
circuit}, with the former being the corresponding wired circuit. The context will 
always make clear whether we are using terms like ``gate'', ``circuit'', etc., 
to refer to physical systems or to their formal abstractions. %

Even if one fully specifies the distinct physical subsystems
of a physical circuit that
will be used to implement each gate in a wired circuit, we
still do not have enough information concerning the physical circuit to
analyze the thermodynamic costs of running it. We still need to specify
the initial states of those subsystems (before the circuit begins running), 
the precise sequence of operations of the gates in the circuit, etc. 
However, before considering these issues, we need to analyze the general form of the thermodynamic costs of running
individual gates in a circuit, isolated from the rest of the circuit. We do that in the next section.

\section{Decomposition of EF}
\label{sec:efdecomp}
Suppose we have a fixed physical system whose dynamics over some time interval is specified
by a conditional distribution $P$, and let $p$ be its initial state distribution, which we can vary.
We decompose the EF of running that system into a sum
of three functions of $p$. Applied to any specific gate in a circuit 
(the ``fixed physical system''), this decomposition tells us how the thermodynamic
costs of that gate  would change
if the distribution of inputs to the gate were changed. 
%

First, \cref{eq:secondlaw} %
tells us that 
the minimal possible EF, across all physical processes %
that transform $p$ into $p' := Pp$, is given by the drop in system entropy.  
%
%
%
%
%
%
%
%
We refer to this drop
as the \textbf{Landauer cost} of computing $\map$ on $p$, 
and write it as %
\begin{align}
\LL(p) := \SSS(p) - \SSS(\map p) \,.
\label{eq:ldef}
\end{align}

Since EF is just Landauer cost plus EP, our next task is to calculate how the EP incurred by a fixed physical
process depends on the initial distribution $p$ of that process. To that end,
in the rest of this section we show that EP can be decomposed into a sum of two non-negative functions of $p$.
Roughly speaking, the first of those two functions reflects the deviation of the initial distribution $p$
from an ``optimal'' initial distribution, while the second term 
reflects the remaining EP that would occur even if the process were run
on that optimal initial distribution. %

To derive this decomposition, we make use of a mathematical result provided by the following theorem.  The theorem considers any function of the initial distribution $p$ which can be written in the form $S(\map p) -S(p) + \mathbb{E}_p[f]$ (i.e., the increase of Shannon entropy plus an expectation of some quantity with respect to $p$). The EP incurred by a physical process can be written in this form (by \cref{eq:newdecomp1}, where $\mathbb{E}_p[f]$ refers to the EF).  Further below, we will also consider other functions, which are closely related to EP, that can be written in this special form. The theorem shows that any function with this special form can be decomposed into a sum of the two terms described above: the first term reflecting deviation of $p$ from the optimal initial distribution (relative to all distributions with support in some restricted set of states, which we indicate as $\sZ$), and a remainder term.

\begin{restatable}{theorem}{thmmaincost}
\label{thm:cost}
Consider any function $\FF :\Delta_\sX \to \mathbb{R}$ of the form
\begin{align*}
\FF(p) := S(\map p) - S(p) + \mathbb{E}_p[f]  
\end{align*}
where  $\map(y\vert x)$ is some conditional distribution of $y \in \sY$ given
 $x\in \sX$ and $f : \sX \to \mathbb{R} \cup \{ \infty \}$ is some function. 
Let $\sZ$ be any subset of $\sX$ such that $f(x)<\infty$ for $x\in\sZ$, and let $q \in \Delta_\sZ$ be any distribution that obeys 
\[
q^\l \in \argmin_{{\rr : \supp \rr \subseteq c}}  \FF(\rr)  \;\;\text{ for all }\;\; \l \in L_\sZ(\map).
\]
Then, each $q^c$ will be unique, and for any $p$ with $\supp p \subseteq \sZ$,
\begin{align*}
\FF(p) & = \DDf{p}{\qPP} - \DDf{\map p}{\map \qPP} + \sum_{\mathclap{\l \in L_\sZ(\map)}} p(c) \FF(\qcPP)  .
\end{align*}
\end{restatable}
\noindent
(We emphasize that that $\map$ and $f$ are implicit in the definition of $\FF$. We remind the reader that the definition of $L_\sZ$ and $q^c$ is provided in \cref{sec:island_def}. See proof in \apprefcost.)

Note that \cref{thm:cost} does not suppose that $\q$ is unique, 
only that the
conditional distributions within each island, $\{\qcPP \}_c$, are. 
Moreover, as implied by the statement of the theorem,
the overall probability weights assigned to the separate islands,  $\{ q(c)\}_c$, has no effect on 
the value of $\FF$. %

%
%
%
%
Consider some conditional distribution $\map(y\vert x)$, with  $\sY=\sX$, implemented by a physical process. Then, if we take $\sZ = \sX$ and $\mathbb{E}_p[f]=\Q$ in \cref{thm:cost}, the function $\FF$
is just the EP of running the conditional distribution $P(y\vert x)$. This establishes the following decomposition of EP:
\begin{align}
\EP(p) = \DDf{p}{\qPP} - \DDf{\map p}{\map \qPP} + \sum_{\mathclap{\l \in \Lmap}} p(c) \EP(q^c) \,.
\label{eq:epdecomp1}
\end{align}
We emphasize that~\cref{eq:epdecomp1} holds without any restrictions on the process, e.g., we do not require that the process obey LDB.
In fact,~\cref{eq:epdecomp1} even holds if the process does not evolve according to a CTMC (as long as EP can be defined via \cref{eq:newdecomp1}).

We refer to the first term in \cref{eq:epdecomp1}, the drop in KL divergence between $p$ and $\qPP$ as both evolve under $\map$, as \textbf{mismatch cost} \footnote{
	\cite{kolchinsky2016dependence} derived \cref{eq:epdecomp1} for the special case 
	where %
	$\map$ has a single island within $\sX$, and only provided a lower bound for more general cases. In that paper, mismatch cost 
	is called the ``dissipation due to incorrect priors'', due to a particular Bayesian interpretation of $q$.
}.
Mismatch cost is non-negative by the data-processing inequality for KL 
divergence~\cite{csiszar_information_2011}.
It equals zero in the special case that $p^c = \qcPP$ 
for each island $c \in L_\sZ(\map)$. %
We refer to any such initial distribution $p$ that results in zero mismatch cost
as a \textbf{prior distribution} of the physical process that implements the conditional distribution $\map$ (the term `prior' reflects a Bayesian interpretation of
$\q$; see~\cite{wolpert_arxiv_beyond_bit_erasure_2015,kolchinsky2016dependence}.) If there is more than one island in $L_\sZ(\map)$, the prior distribution is not unique. 

We call the second term in our decomposition of EP in \cref{eq:epdecomp1}, $\sum_{\l \in L_\sZ(\map)} p(c) \EP(q^c)$,
the \textbf{residual EP}. 
In contrast to mismatch cost, residual EP
does not involve information-theoretic quantities, and depends linearly on $p$.
When $L_\sZ(\map)$ contains a single island, %
this ``linear'' term reduces to an additive constant, %
independent of the initial distribution. %
The residual EP terms $\{ \sigma(q^c)\}_c$ are all non-negative, since EP is non-negative.
%
%
%
%
%
%
%
%
%
%
%
%

Concretely, the conditional distributions $\{\qcPP\}_c$ 
and the corresponding 
set of real numbers $\{\EP(\qcPP)\}_c$ depend on the precise physical details 
of the process, beyond the fact that that process implements
$\map$.  Indeed, by appropriate design of the ``nitty gritty'' details of the physical process,
it is possible to have $\EP(\qcPP) = 0$ for all $\l \in L_\sZ(\map)$, in which case 
the residual EP would equal zero for all $p$. (For example, this will be the
case if the process is an appropriate quasi-static transformation; see~\cite{owen_number_2018,esposito2010finite}.)
 
Imagine that the conditional distribution $\map$ is logically reversible over some set of states $\sZ\subseteq \sX$, and that $\supp p \subseteq \sZ$. Then, both mismatch cost and
Landauer cost must equal zero, and EF must equal EP, which in turn must equal residual EP~\footnote{When \unexpanded{$\map$} is logically reversible over initial states $\sZ$, each state in \unexpanded{$\sZ$} is a separate island, which means that 
EF, EP, and residual EP can be written as \unexpanded{$\sum_{x \in \sZ} p(x) \EP(u^x)$}, where \unexpanded{$u^x$} 
indicates a distribution which is a delta function over state \unexpanded{$x$}.}.
Conversely, if $\map$ is not logically reversible over $\sZ$, then mismatch cost cannot 
be zero for all initial distributions $p$ with $\supp p \subseteq \sZ$ (for such a $P$, regardless of what $q$ is, there will be some $p$ with $\supp p \subseteq \sZ$ such that the KL divergence between $p$ and $q$ will shrink under the mapping $\map$). 
Thus, for any fixed process that implements a logically irreversible map, there will be some initial distributions $p$ that result in unavoidable EP.

\def\peq{p_\mathrm{eq}}

To provide some intuition into these results,
the following example reformulates the EP of a very commonly considered scenario as a special case of \cref{eq:epdecomp1}:
\begin{example}
Consider a physical system evolving according to an irreducible master equation, %
while coupled to a single thermodynamic reservoir and without external driving.  
Because there is no external driving, the master equation is time-homogeneous with some %
unique equilibrium distribution $\peq$. %
So the system is relaxing toward that equilibrium as it undergoes the conditional distribution $P$ over the interval
$t \in [0, 1]$.

For this kind of relaxation process, it is well known that the EP can be written as~\cite{schlogl1971stability,schnakenberg_network_1976,esposito2010three}:
\begin{align}
\label{eq:eprelax}
\EP(p) = \DDf{p}{\peq} - \DDf{\map p}{\peq} \,.
\end{align}
\cref{eq:eprelax} can also be derived from our result,
 \cref{eq:epdecomp1}, since 
\begin{enumerate}
\item Taking $\sZ=\sX$, $P$ has a single island
(because the master equation is irreducible, and therefore any state %
 is reachable from any other over $t\in[0,1]$);
\item  The prior distribution within this single island is $q =\peq$ (since the EP
would be exactly zero if the system were started at this equilibrium, which is a fixed point of $\map$); 
\item The residual EP is $\sigma(q) = 0$ (again using fact that
EP is exactly zero for $p = \peq$, and that there is a single island);
\item  $\map q =\peq$ (since there is no driving, and the equilibrium distribution is a fixed point of $\map$). %
\end{enumerate}
Thus, \cref{eq:epdecomp1} can be seen as a generalization of the well-known relation given by \cref{eq:eprelax}, which is defined for simple relaxation processes, to processes that are driven and possibly connected to multiple reservoirs.  
\end{example}

The following example addresses the effect of possible discontinuities in the island decomposition of $\map$
on our decomposition of thermodynamic costs:

\begin{example}
	
Mismatch cost and residual EP are both defined in terms of the island decomposition of the conditional distributions $\map$ over some set of states $\sZ$. 
That decomposition in turn depends on which (if any) 
entries in the conditional probability distribution $\map$ are exactly 0.  
%
%
This suggests that the decomposition of \cref{eq:epdecomp1} 
can depend discontinuously on very small 
variations in $\map$ which replace strictly zero entries in $\map$ with infinitesimal values, since
such variations will change the island decomposition of $\map$.
%
%
%
%
%

To address this concern, first note that if $\map \simeq\map'$, then the EP of the real-world process that 
implements $\map'$ can be approximated as
\begin{align}
\EP'(p) & = S(\map' p)- S(p) + \Q'(p) \nonumber \\
& \simeq  S(\map p) - S(p) + \Q'(p) ,
\label{eq:approx0}
\end{align}
where $\Q'(p)$ is the EF function of the real-world process, with the approximation becoming exact as $\map \rightarrow \map'$
\footnote{The fact that
	$S(\map' p)\rightarrow  S(\map p)$ as $\map \rightarrow \map'$ follows from
	\cite[Thm. 17.3.3]{cover_elements_2012}.}. If we now apply \cref{thm:cost} to the RHS of \cref{eq:approx0}, we see that so long as $\map'$ is close enough to $\map$,
we can approximate
$\EP'(p)$ as a sum of mismatch cost and residual EP using the islands 
of the idealized map $\map$, instead of the actual map $\map'$.
\end{example}

\section{Solitary processes}
\label{sec:solitary}

Implicit in the definition of a physical circuit is that it is ``modular'', in 
the sense that when a gate in the circuit runs, 
it is physically coupled to the gates that are its direct inputs, and those that directly
get its output, but is %
\textit{not} physically coupled to any other gates in the circuit.
This restriction on the allowed physical coupling is a constraint on the possible processes that implement each gate 
in the circuit. It has major thermodynamic consequences, which we analyze in this section. 

To begin, suppose we have a system that can be decomposed into two 
separate subsystems, $A$ and $B$, so that the system's overall state space $\sX$ can be written as
$\sX = \sXAB$, with states $(\xA, \xB)$. 
For example, $A$ might contain a particular gate and its inputs, while $B$ might consist
of all other nodes in the circuit.  %
We use the term {\textbf{solitary process}} to refer to a physical process over state space $\sXAB$ 
that takes place during $t\in[0,1]$ where: %
\begin{enumerate}
\item  %
 $A$ evolves independently of $B$, and $B$ is held fixed:
\begin{align}
P(x_A', x_B' \vert x_A, x_B) = \PA(x_A'\vert x_A)\delta(x_B', x_B) \,.
\label{eq:fracsystem1}
\end{align}
\item 
The EF of the process depends only on the initial distribution over $\mathcal{X}_A$, which we indicate with the following notation: 
%
%
%
\begin{align}
\label{eq:fracsystem2}
\Q(p)=\Q_A(\pA) \,.%
\end{align}
\item The EF %
is lower bounded by the change in the marginal entropy of subsystem $A$,
\begin{align}
\label{eq:fracsystem3}
\Q_A(\pA) &\ge S(\pA) - S(\PA \pA) \,.
\end{align}
\end{enumerate}

Note that it may be that some subset $A'$ of
the variables in subsystem $A$ don't change their state during the solitary process. In that sense 
such variables would be
like the variables in $B$.  However, if the dynamics of those variables in $A$ that \textit{do} change
state depends on the values of the variables in $A'$, 
then in general the variables in $A'$ cannot be assigned to $B$; they have
to be included in subsystem $A$ in order for condition (2) to be met.

\begin{example}
\label{ex:1}
A concrete example of a solitary process is
a CTMC where at all times, the rate matrix $\Wt$  %
has the decoupled structure
\begin{align}
\Wt(x_V  \totrans  x_V') = \delta(x_B, x_B') \sum_\bathix\Wt^{A,\bathix}(x_A  \totrans  x_A')\label{eq:decoupledW}
\end{align}
for $x_V \ne x_V'$, where $K_t^{A,\bathix}$ indicates the rate matrix for subsystem $A$ and thermodynamic reservoir $\bathix$ at time $t$~\footnote{In fact, in~\cite{wolpert_thermo_comp_review_2019} solitary processes are defined as CTMCs with this form.}. 

To verify that this CTMC  is a solitary process, first plug 
 the rate matrix in \cref{eq:decoupledW} into \cref{eq:ctmc_dynamics2} and simplify, giving 
\begin{align*}
\frac{d}{dt} p^t(x'_A, x'_B) = 
p^{t}(x_{B}')\sum_{x_{A}}p^{t}(x_{A}\vert x_{B}')\sum_\bathix \Wt^{A,\bathix}(x_{A}\to x_{A}').
\end{align*}
Marginalizing the above equation, we see that 
the distribution over the 
states of $A$ evolves independently according to
\begin{align*}
{\frac{d}{dt}} p^t(x_A') = %
\sum_{x_A} p^t(x_A) \sum_\bathix \Wt^{A,\bathix}(x_A  \totrans  x_A') \,.
\end{align*}
Note also that, given the form of \cref{eq:decoupledW}, the state of $B$ does not change. Thus, the conditional distribution carried out by this CTMC over any time interval must have the form of \cref{eq:fracsystem1}. (See also App.~B in~\cite{wolpert_thermo_comp_review_2019}.)

Next, plug \cref{eq:decoupledW} into \cref{eq:defefrate0} and simplify to get
\begin{align}
\dot{\Q}(p^t) = \sum_{\mathclap{\bathix, x_A,x_A'}} p^t(x_A) \Wt^{A,\bathix}(x_A  \totrans  x_A') \ln \frac{\Wt^{A,\bathix}(x_A  \totrans x_A')}{\Wt^{A,\bathix}(x_A'  \totrans  x_A)} \,.
\label{eq:22a}
\end{align}
Thus, the EF incurred by the process evolves exactly as if $A$ were an independent system connected to a set of thermodynamic reservoirs. %
Therefore, a joint system evolving according to \cref{eq:decoupledW} 
will satisfy \cref{eq:fracsystem2,eq:fracsystem3}. 
\end{example}

We refer to the lower bound on the EF of subsystem $A$, as given in \cref{eq:fracsystem3}, 
as the \textbf{subsystem Landauer
cost} for the solitary process. We make the associated definition
that the \textbf{subsystem EP} for the solitary process is
\begin{align}
\subEP_A(p_A) := \Q_A(p_A) - \big[\SSS(p_A) - \SSS(\map_A p_A)\big] ,
\end{align}
which by \cref{eq:fracsystem3} is non-negative. 
Note that if $\map_A$ is a logically reversible conditional distribution, then subsystem EP is equal to the EF incurred by the solitary process.

In general,
$\SSS(p_A) - \SSS(\map_A p_A)$, the subsystem Landauer cost, will not 
equal $\SSS(p_{AB}) - \SSS(\map p_{AB})$, the Landauer cost of the  entire joint system. 
Loosely speaking, an observer examining the entire system would ascribe a different value to its 
entropy change during the solitary process than would an observer examining just subsystem $A$ --- even though
subsystem $B$ doesn't change its state. We use the term  \textbf{Landauer loss} to refer to  this difference in Landauer costs,
\begin{align}
\label{eq:land-loss-def}
\LandauerLoss(p) &:= 
\big[\SSS(p_A) - \SSS(\map_A p_A)\big] 
- \big[\SSS(p_{AB}) - \SSS(\map p_{AB})\big] .
\end{align}
Assuming that the lower bound  \cref{eq:fracsystem3} can be saturated,
since the bound \cref{eq:secondlaw} can be saturated, 
the Landauer loss is the increase in the minimal EF that must be incurred by any process that carries out $\map$
if that process is required to be a solitary process.

By using the fact that subsystem $B$ remains fixed throughout a solitary process, the Landauer loss can  be rewritten as the drop  in the mutual information between $A$ and $B$, from the beginning to the end of the solitary process,
\begin{align}
\LandauerLoss(p) &= I_p(A;B) - I_{\map p}(A;B) .
\label{eq:drop_in_mutual}
\end{align}
Applying the data processing inequality establishes that Landauer loss is non-negative~\cite{cover_elements_2012}.
(See \cref{sec:earlierwork} for a discussion of the relation between solitary processes and other processes that have been
considered in the literature.)

If $\map_A$ (and thus also $\map$) is logically reversible, then the Landauer loss will always be
zero. However, for other conditional distributions, there is always some $p$ that results in strictly positive Landauer loss.
Moreover, we can rewrite it as
\begin{align}
	\LandauerLoss(p) = 
	\EP(p_{AB}) - \subEP_A(p_{A}) .
	\label{eq:lldiff}
\end{align}
So in general the subsystem EP will be less than the overall EP of the
entire system
%
%
%
%
%
%
%
%
%
%
\footnote{See~\cite{wolpert_thermo_comp_review_2019} for an example explicitly
	illustrating how the rate matrices change if we go from 
	an unconstrained process that implements $P_A$ to a solitary process that does so, and how that change increases the total EP.
	That example also considers the special case where the prior of the full $A \times B$ system 
	is required to factor into a product of a distribution over the initial
	value of $x_A$ times a distribution over the initial value of $x_B$. In particular, it shows that the Landauer loss is the minimal value of the mismatch cost 
	in this special case.}.

Finally, note that $\Q_A(p_A)$ is a linear function of the distribution $p_A$ (since EF functions are linear). %
Combining this fact with
\cref{thm:cost}, while taking $\sZ=\sX_A$, allows us to expand the subsystem EP as 
\begin{multline}
\label{eq:subep}
\subEP_A(p_A)  =\DDf{p_A}{\qPP_A} - \DDf{\map_A p_A}{\map_A \qPP_A} + \\
        \sum_{\mathclap{\l \in \LLL(\map_A)}} p_A(c) \subEP_A(\qPP_A) ,
\end{multline}
where $\qPP_A$ is a distribution over $\sX_A$ that satisfies $\subEP_a(\qPP_A^c) = \min_{r : \supp r \subseteq c} \subEP_A(r)$ for all $\l \in \LLL(\map_A)$.
As before, both the drop in KL divergence and the term linear in $p_A(c)$ are
non-negative.  
We will sometimes refer to that drop in KL divergence as
\textbf{subsystem mismatch cost}, with $q_A$ the \textbf{subsystem prior}, and refer to the linear term as \textbf{subsystem residual EP}.
Intuitively, subsystem Landauer cost, subsystem EP, subsystem mismatch cost, and subsystem residual EP are simply
the values of those quantities that an observer would ascribe to subsystem $A$
if they observed it independently of $B$. %

\section{Serial-Reinitialized circuits}
\label{sec:singlepass}

\label{sec:how_a_gate_runs}

As mentioned at the end of \cref{sec:circuit_theory}, specifying a wired circuit does not specify 
the initial distributions of the gates in the physical circuit, the sequence in which the gates 
in the physical circuit are run, etc. So it does not fully specify the dynamics  
of a physical system that implements that wired circuit.
In this section we introduce one relatively simple way of mapping a wired circuit to 
such a full specification. 
In this specification, 
%
the gates are run serially, one after the other. Moreover, the gatesreinitialize the states
of their parent gates after they run, so that the entire circuit can be repeatedly run, incurring the same expected 
thermodynamic costs each time. We call such physical systems \textbf{serial reinitialized implementations} of a given
wired circuit, or just {SR circuits} for short.

For simplicity, 
in the main text of this paper we focus on the special case 
in which all non-output nodes have out-degree 1, i.e., where 
each non-output node is the parent of exactly one gate. 
See {\apprefoutdegreebig} for a discussion of how to extend the current analysis
to relax this requirement, allowing some nodes to have out-degree larger than 1. 

There are several properties that jointly define the {SR circuit} 
implementation of a given wired circuit.

First, just before the physical circuit starts to run, 
all of its nodes have a special initialized value with probability $1$, i.e., 
$\variablevalue_v = \nullstate$ for all $v \in V$ at time $t=0$. %
Then the joint state of
the input nodes $\variablevalue_\Cin$ is set by sampling $\pIN(x_\Cin)$ \footnote{Strictly speaking, if
the circuit is a Bayes net, then $\pIN$ should be a product distribution over the root nodes.
Here we relax this requirement of Bayes nets, and let  $\pIN$ have arbitrary correlations.}. 
Typically this setting of the state of
the input nodes is done by some offboard system, e.g., the user of the digital device containing the circuit.
We do not include the details of this offboard system in our model of the physical circuit. Accordingly, we
do not include the thermodynamic costs of 
setting the joint state of the input nodes in our calculation of the thermodynamic 
costs of running the circuit \footnote{For example, it could be that at some $t<0$, 
the joint state of the input nodes is some special initialized state $\vec{\nullstate}$ with probability $1$,
and that that initialized joint state is then overwritten with the values copied in from
some variables in an offboard system, just before the circuit starts. The joint entropy of the offboard
system and the circuit would not change in this overwriting operation, and so 
it is theoretically possible to perform that operation with zero EF~\cite{parrondo2015thermodynamics}. 
However, to be able to run the circuit again after it finishes, with new
values at the input nodes set this
way, we need to reinitialize those input nodes to the joint state $\vec{\nullstate}$ .
As elaborated below, we \textit{do} include the thermodynamic costs of reinitializing those input
nodes in preparation of the next run of the circuit.
This is consistent with modern analyses of Maxwell's demon, 
which account for the costs of reinitializing the demon's memory 
in preparation for its next run~\cite{parrondo2015thermodynamics,wolpert_thermo_comp_review_2019}.
See also \cref{foot:series_of_circuits}.
\label{foot:cleanup}
}. 

After $x_\Cin$ is set this way, the SR circuit implementation begins. It works
by carrying out a sequence of solitary processes, one for each gate of the circuit,
including wire gates. 
At all times that a gate $g$ is ``running'', the combination of that gate and its parents (which we indicate as $\SG$) is the subsystem $A$ in the 
definition of solitary processes.  
The set of all other nodes of the wired circuit ($V\setminus \SG$) constitute the
subsystem $B$ of the solitary process.  
The temporal ordering of the solitary processes must be a topological ordering
consistent with the wiring diagram of the circuit: 
if gate $g$ is an ancestor of gate $g'$, then the solitary process for gate $g$ 
completes before the solitary process for gate $g'$ begins.

When the solitary process corresponding to any gate $g \in G$ begins running, $\variablevalue_g$ is still set to its initialized
state, $\nullstate$, while all of the parent nodes of $g$ are either input nodes, or other gates that have 
completed running and are set to their output values. %
By the end of the solitary process for gate $g$, %
$\variablevalue_g$ is set to a random sample of the
conditional distribution $\mG(x_g \vert x_\pag)$, while its 
parents are reinitialized to state 
$\nullstate$. %
More formally, under the solitary process for gate $g$, nodes $\SG$ evolve according to 
\begin{align}
P_g(x_{\SG}'\vert x_{\SG}) :=
\mG(x_g' \vert x_\pag) \prod_{v \in\pag} \delta(x_{v}', \nullstate) %
\label{eq:spmap}
\end{align}
while all nodes $V \setminus \SG$ do not change their states. 
(Recall notation from \cref{sec:circuit_theory}.) Note that this means that the input nodes are reinitialized
as soon as their child gates have run.

\begin{example}
In this example we demonstrate how to implement an XOR gate $g$ in an SR circuit with a CTMC, i.e., how to carry out the following logical map on the state of gate $g$,
\[
\mG(x_g\vert x_\pag) =\delta(x_g,\mathsf{XOR}(x_\pag)). 
\]
The CTMC involves a sequence of two solitary processes over $\SG$. The time-dependent rate matrix for both solitary processes has the form
\begin{align*}
\Wt(x_V \totrans x_V')=\delta(x_{V\setminus \SG}, x_{V\setminus \SG}') \Wt^{\SG}(x_\SG  \totrans  x_\SG')
\end{align*}
for all $x_V \ne x_V'$ (compare to \cref{eq:decoupledW}, where for simplicity we assume there is a single thermodynamic reservoir). The two solitary processes differ in their associated subsystem rate matrices  $\Wt^{\SG}$. 

In the first solitary process, the state of the gate's parents is held fixed, while the gate's output is changed from the initialized state to the correct XOR value. 
For $t\in [0,1]$ (the units of time are arbitrary), the subsystem rate matrix that implements this solitary process is 
\begin{multline}
\Wt^\SG \big(x_\SG \totrans x'_\SG \big) = \delta(x_{\pag}, x'_{\pag}) \times \\
\eta \left[ (1-t)\delta(x_g',\nullstate)/4 + t \mG(x_g'\vert x_\pag')/4 \right],
\label{eq:exrm}
\end{multline}
for $x_\SG \ne x'_\SG$, where $\eta > 0$ is the relaxation speed.  
Note that the term $\delta(x'_\SG, \nullstate)$ inside the square brackets encodes the assumption that the initial state of the gate is $\nullstate$ with probability $1$, while the factor of $1/4$ encodes the assumption that 
the initial distribution over the four possible states of the gate's parents is uniform.

From the beginning to the end of the first solitary process, the nodes $\SG$ are updated according to the conditional probability distribution $\map_g^{(1)}$, given by the time-ordered exponential of the rate matrix in \cref{eq:exrm} over $t\in[0,1]$. 
In the quasi-static limit $\eta \to \infty$, this conditional distribution becomes
\begin{align*}
\map_g^{(1)}(x_\SG'\vert x_\SG) = \delta(x_\pag,x_\pag')\mG(x_g'\vert x_\pag).
\end{align*}

In the second solitary process, the gate's output is held fixed while the gate's parents are reinitialized.  
Redefining the time coordinate so that this second process also transpires in $t\in [0,1]$, its subsystem rate matrix is 
\begin{align}
&\Wt^\SG \big(x_\SG \totrans x'_\SG \big) = \delta(x_{g}, x'_{g}) \times  \label{eq:exrm2} \\
& \quad \eta \bigg[(1-t)\mG(x_g'\vert x_\pag')/4 + t \prod_{\mathclap{v\in\pag}} \delta(x_v', \nullstate)/2 \bigg], \nonumber
\end{align}
for $x_\SG \ne x'_\SG$, where $\eta$ is again the relaxation speed. 
Note that $\mG(x_g'\vert x_\pag')/4$ is what the distribution over nodes $\SG$ would be at the beginning of the second solitary process, if the distribution at the beginning of the first solitary process was $\delta(x_g',\nullstate)/4$. 
From the beginning to the end of the second solitary process, the nodes $\SG$ are updated according to the conditional probability distribution $\map_g^{(2)}$, which is given by the time-ordered exponential of the rate matrix \cref{eq:exrm2}. In the quasi-static limit $\eta \to \infty$, this conditional distribution is
\begin{align*}
\map_g^{(2)}(x_\SG'\vert x_\SG) = \delta(x_g,x_g')\prod_{\mathclap{v \in\pag}} \delta(x_v', \nullstate).
\end{align*}

The sequence of two solitary processes causes the nodes in $\SG$ to be updated according to the conditional distribution  $P_g = \map_g^{(1)} \map_g^{(2)}$.  In the quasi-static limit, this is  
\begin{align}
\map_g(x_\SG'\vert x_\SG)  =\mG(x_g'\vert x_\pag) \prod_{\mathclap{v\in\pag}} \delta(x_v', \nullstate),
\label{eq:PgXOR}
\end{align}
which recovers \cref{eq:spmap}, as desired.

We now compute thermodynamic costs for the XOR gate. Let $\Q(\ppag)$ be the total EF incurred by running the sequence of two solitary process, given some initial distribution $\ppag$ over the parents of gate $g$.  Using results from \cref{sec:solitary}, write this EF as
\begin{align}
\Q(\ppag) =&\; S(\ppag) - S(\mG p_{\pag}) \nonumber \\
&+\DDf{\ppag}{\qpag} - \DDf{\mG \ppag}{\mG \qpag} \nonumber \\
&+\sum_{\mathclap{\l \in \LLL(\mG)}} \ppag(c) \subEP_\SG(\qpag),
\label{eq:exampe_full_decomp}
\end{align}
where the three lines correspond to subsystem Landauer cost, subsystem mismatch cost, and subsystem residual EP, respectively.  To derive this decomposition, we
 applied \cref{thm:cost}, while taking $\sZ = \{  x_\SG \in \sX_\SG : x_g = \nullstate \}$ (note that for this $\sZ$, $L_\sZ(\map) = \LLL(\mG)$). 

To compute the second and third of those terms, note that
in the quasi-static limit, the prior distribution is uniform:
\begin{align}
\qpag(x_\pag) = 1 / 4 .
\label{eq:prior1}
\end{align}
To see this,
suppose that the distribution over $\SG$ when the sequence of processes begins is given by $p_\SG(x_\SG) = \delta(x_g,\nullstate)\qpag(x_\pag)$.  Then,
\begin{enumerate}
 \item The system will remain in equilibrium during the first solitary process, thereby incurring zero EP. At the end of the first solitary process, it will have distribution
 \begin{align}
 [\map_g^{(1)} p_\SG](x_\SG) = \qPP_\pag(x_\pag)\mG(x_g'\vert x_\pag) .
 \label{eq:prior2}
 \end{align}
 \item Given that the system starts the second solitary process with this distribution $\map_g^{(1)} p_\SG$, it will remain in equilibrium throughout the second solitary process, thereby again incurring zero EP.
\end{enumerate}
So that sequence of processes will incur zero EP
 --- the minimum possible --- if the initial distribution is $\qpag$ over $\pag$ (and $x_g=\nullstate$), as claimed. 
In addition, the fact that the minimal EP that can be generated for any initial distribution is strictly zero means that the subsystem residual EP vanishes. This fully specifies all terms in \cref{eq:exampe_full_decomp}, as a function of $\ppag$.

As a concrete example of this analysis, consider the initial distribution which is uniform over states $\{\mathsf{00},\mathsf{01},\mathsf{10}\}$:
\[
\ppag(x_\pag) = [1-\delta(x_\pag, \mathsf{11})]/3 .
\]
For this distribution, the subsystem Landauer cost is
\begin{multline*}
S(\ppag) - S(\mG \ppag) = \\
\ln 3 + [(1/3)\ln(1/3) + (2/3)\ln(2/3)] \approx 0.46 .
\end{multline*}
The subsystem EP is 
\begin{multline*}
\DDf{\ppag}{\qpag} - \DDf{\mG \ppag}{\mG \qpag} = \\
 [\ln 4 - \ln 2] - [S(\ppag) - S(\map_g \mG \ppag)] \approx 0.23 .
\end{multline*}

We end by noting that the XOR gate may also incur some EP which is not accounted for by these calculations, due to loss of  correlations between the nodes $\SG$ and the rest of the circuit as the gate runs. This is quantified by the Landauer loss, which can be evaluated using \cref{eq:land-loss-def}, \cref{eq:drop_in_mutual}, or \cref{eq:lldiff}.

\label{ex:new_3}
\end{example}

\begin{figure*}
  \includegraphics[width=1\linewidth]{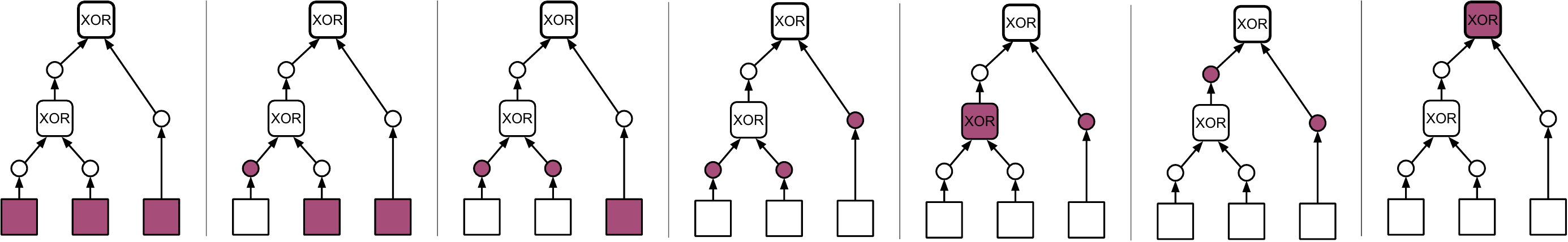}
  \caption{An SR implementation of the wired circuit shown in \cref{fig:wiredcircuit}. Each diagram represents one step of the SR implementation, with white shapes indicating nodes set to their initialized value ($\nullstate$) and maroon shapes indicates nodes that can have non-initialized values. The implementation starts with only the input nodes set to non-initialized values (left-most diagram) and ends with only the output gates set to non-initialized values (right-most diagram).} 
  \label{fig:srflow}
\end{figure*}

Given the requirement that the solitary
processes are run accordingly to a topological ordering, \cref{eq:spmap} ensures that once all the gates of the circuit have
run, the state of the output gates
of the circuit have been set to a random sample of $\mC(x_\Cout|x_\Cin)$, while 
all non-output nodes are back in their
initialized states, i.e.,  $\variablevalue_v = \nullstate$ for $v \in V \setminus \Cout$. %
%

\begin{example}
Consider the 3-bit parity circuit shown in \cref{fig:wiredcircuit}.  An SR implementation of this wired circuit would run its 6 gates in topological order, such that each gate computes its output and then reinitializes its parents.  One such sequence of steps is shown in \cref{fig:srflow} (note that some other topological orderings are also possible). Each XOR gate could be implemented by the kind of CTMC described in \cref{ex:new_3}. Each wired gate could be run by a similar kind CTMC, but which carries out the identity operation $\mG(x_g'\vert x_\pag) = \delta(x_g',x_\pag)$, instead of the XOR operation.
\end{example}

After the SR circuit has run,
some ``offboard system'' may make a copy of 
the state of the output gate for subsequent use, e.g., by copying it into some of 
the input bits of some downstream circuit(s), onto an external disk, %
etc. Regardless, we assume that after the circuit finishes, 
but before the circuit is run again, the state of the output nodes have also
been reinitialized to $\nullstate$.
Just as we do not model the physical mechanism by which new
inputs are created for the next run of the circuit, we also do not
model the physical mechanism by which the output of the circuit is
reinitialized. Accordingly, in our calculation of the thermodynamic costs of running the 
circuit, we do not account for any possible cost of reinitializing the output~\footnote{Suppose that the outputs of circuit $\CC$ were the \textit{inputs} of
some subsequent circuit $\CC'$. That would mean that when $\CC'$ reinitializes its inputs,
it would reinitialize the outputs of $\CC$. Since we ascribe the 
thermodynamic costs of that reinitialization to $\CC'$, it would result in
double-counting to also ascribe the costs of reinitializing $\CC$'s outputs to $\CC$.
\label{foot:series_of_circuits}
}.

%
%
%
This kind of cyclic procedure for running the circuit 
allows the circuit to be re-used an arbitrary number of times, while
ensuring  
that each time it will %
have the same expected thermodynamic
behavior  (Landauer cost, mismatch cost, etc.), and will carry out the same map $\mC$ from input nodes to output gates.

\section{Thermodynamic costs of SR circuits}
\label{sec:singlepasscosts}

In general, there are multiple decompositions of the EF and EP incurred by running any given SR circuit.
They differ in how much of the detailed structure of the circuit they incorporate. In this section we
present some of these decompositions.
(See \apprefsinglepass for all proofs of the results in this section.)

\subsection{General decomposition of EF and EP}

Let $p$ refer to the initial distribution over the joint state of all nodes in the circuit. By \cref{eq:newdecomp1}, 
the total EF incurred by implementing some overall map $P$ which takes the initial joint state of \textit{all} nodes
in the the full circuit to the final joint state is
\begin{align}
\Q(p) = \LL(p) + \EP(p) .
\label{eq:zzz99}
\end{align} 
where
\begin{align}
\LL(p) &:= \SSS(p) - \SSS(P p) 
\label{eq:spl0}
\end{align}
is the Landauer cost of computing $P$ for initial distribution $p$. 

The first term in \cref{eq:spl0}, $\LL$, is the minimal EF that must be incurred by any process over
$\sX_\Cin \times \sX_\Cout$ that implements $\map$, without any constraints on how
the variables in $\sX_\Cin \times \sX_\Cout$ are coupled, 
and without any reference to a set of intermediate subsystems
(e.g., gates) that may connect the input and output variables~\footnote{See~\cite{wolpert_thermo_comp_review_2019}
for a discussion of how this bound applies in the case of ``logically reversible circuits''.}.

The second term in \cref{eq:zzz99} is the EP, which reflects the thermodynamic irreversibility of the SR circuit. Using  \cref{eq:epdecomp1}, the EP can be further decomposed as 
\begin{align}
\EP(p) = \big[ \DDf{p}{\qPP} - \DDf{\map p}{\map \qPP} \big] + \sum_{\mathclap{\l \in \Lmap}} p(c) \EP(\qPP^c)  .
\label{eq:systemwide1}
\end{align}
The decrease in KL reflects the mismatch cost, arising from the discrepancy between $p(x_V)$, the actual 
initial distribution over {all} nodes of the circuit defined in \cref{eq:p_IN}, 
and $q(x_V)$, the optimal prior distribution over the joint state of all the nodes of the circuit which would result in the least EP.  The last sum in \cref{eq:systemwide1} reflects the residual EP, reflecting EP that remains even when the circuit is initialized with the optimal prior distribution.

Suppose we know that the dynamics is actually implemented with an SR circuit, but don't know
the precise wiring diagram. Then we
know that the initial joint distribution over all the nodes is
\begin{align}
p(x_V) = \pIN(x_\Cin) \prod_{v \in V \setminus \Cin} \delta(x_v, \nullstate) \,,
\label{eq:p_IN}
\end{align}
and the ending joint distribution is
\begin{align}
[Pp](x_V) = 
p_\Cout(x_\Cout) \prod_{v \in V \setminus \Cout} 
\delta(x_v, \nullstate) \,,
\end{align}
So $S(p) = S(\pIN)$ and 
$\SSS(P p) = \SSS(p_\Cout)= \SSS(\mC \pIN)$, 
where $\mC$ is the conditional distribution of the final joint state of the output gates given the initial joint state of the input nodes, defined in \cref{eq:defmc}.  
Combining gives
\begin{align}
\LL(p) & = \SSS(\pIN) - \SSS(\mC \pIN) 
\label{eq:spl}
\end{align}
Similarly, the EP becomes
%
\begin{multline}
\EP(p) = \big[ \DDf{\pIN}{\ppIN} - \DDf{\mC \pIN}{\mC \ppIN} \big] + \\\sum_{\mathclap{\l \in \LLL(\mC)}} p(c) \EP(\ppIN^c) .
\label{eq:systemwide2}
\end{multline}

While the expressions in \cref{eq:systemwide1,eq:systemwide2} for EP must be equal, how they decompose that EP among a mismatch cost
term and a residual EP differ. 
The two decompositions differ because they define the ``optimal initial distribution'' relative to differ sets of possible distributions, resulting in different prior distributions (which are also defined over different sets of outcomes). 
Also note that the residual EP terms in \cref{eq:systemwide2} are defined in terms of a more constrained minimization problem than the residual EP terms in  \cref{eq:systemwide1}. Thus, given the same initial distribution $p$, the residual EP  in \cref{eq:systemwide2} will generally be larger than the residual EP in \cref{eq:systemwide1}, while the mismatch cost in \cref{eq:systemwide1} will generally be larger than the mismatch cost in \cref{eq:systemwide2}.
We also emphasize that the  island decompositions appearing in the two expressions are different.

\subsection{Circuit-based decompositions of EF and EP}
\label{sec:ef_decomp_circuit}
\label{sec:alt_decomp_EP}

The decompositions of EF and EP given in
(\cref{eq:zzz99,eq:systemwide1,eq:systemwide2}) do not involve the wiring diagram of the SR circuit.  As an alternative,
we can exploit that wiring diagram to formulate a decomposition of EF and EP which separates the contributions from different gates. In general, such circuit-based decompositions allow for a finer-grained analysis of the EP in SR circuits than do
the decompositions proposed in the last section.  In particular, 
they allow us to derive some novel connections between nonequilibrium statistical
physics, computer science theory, and information theory, as discussed in the next two subsections.  


Before discussing these circuit-based decompositions, we introduce some new notation. 
We write $\ppag(x_{\pag})$ and 
$\pSG(x_\SG) = \ppag(x_{\pag}) \delta(x_g, \nullstate)$ for the distributions over
$x_\pag$ and $x_\SG$, respectively, 
at the beginning of the solitary process that implements gate $g$. 
We write the EF function of the solitary process of gate $g$ as $\Q_g(\pSG)$, and its subsystem EP as
\begin{align}
\EPg(\pSG) := \Q_g(\pSG) - [S(\pSG) - S(P_g \pSG)] \,.
\label{eq:def_sub_EP}
\end{align}
We also write $p^{{\text{beg}}(g)}$ and $p^{{\text{end}}(g)}$ to indicate the joint
distribution over \textit{all} circuit nodes at the beginning and end, respectively, 
of the solitary process that runs gate $g$. As an illustration of this notation, $p^{{\text{beg}}(g)}(x_\pag) =
p_{\pag}(x_{\pag})$. On the other hand, $p^{{\text{end}}(g)}(x_g) = (\mG p_\pag)(x_g)$ is the distribution over $x_g$ 
after gate $g$ runs, and  $p^{{\text{end}}(g)}(x_{\pag})$ is a delta function about the joint state of
the parents of $g$ in which they are all initialized, by \cref{eq:spmap}. Note that since we're considering solitary processes,
$p^{{\text{beg}}(g)}(x_{V \setminus \SG}) = p^{{\text{end}}(g)}(x_{V \setminus \SG})$.

We now present our first circuit-based decomposition, and then we
explain what its terms mean in detail:
\begin{theorem}
The total EF incurred by running an SR circuit where $p$ is the initial distribution over the joint state of all nodes in the circuit  is
\begin{align}
\Q(p) = \LL(p) + \underbrace{\LandauerLoss(p) + \MM(p)+ \RR(p)}_{\EP(p)}  .
\label{eq:decomp9}
\end{align}
\label{thm:maindecomp}
\end{theorem}

\vspace{5pt}

\noindent \textbf{{1)}}
The first term in \cref{eq:decomp9}, $\LL(p)$, is the Landauer cost of the circuit, as described in \cref{eq:spl0}.  This Landauer cost can be further decomposed into contributions from the individual gates. Specifically,  write $\LL_g(p)$ 
for the drop in the entropy of the entire circuit during the time that the solitary process
for gate $g$ runs, given that the input distribution over the entire circuit is $p$:
\begin{align}
\LL_g(p) := S(p^\begG)- S(p^\ennG) \,.
\label{eq:system_wide_land_gate}
\end{align}
Note that the distribution over the states of the entire physical circuit at the {\it{end}} of
the running of any gate is the same as the distribution at the{ \it{beginning}}
of the running of the next gate. 
So by canceling 
terms, and using the fact that entropy does not change when a wire
gate runs, we can expand $\LL$ as %
\begin{align}
\LL (p) = \sum_{g \in G \setminus W} \LL_g(p) \,.
\label{eq:land_AO}
\end{align}
(Recall from \cref{sec:circuit_theory} that $W$ is the set of wire gates in the circuit.)
This decomposition will be useful below.
\vspace{5pt}

\noindent \textbf{{2)}}
The second term in \cref{eq:decomp9}, $\LandauerLoss(p)$, is the unavoidable 
additional EF that is incurred by {any} SR implementation of the SR circuit on initial distribution $p$, above and beyond 
$\LL$, the Landauer cost of running the map $\mC$ on initial distribution $p$.  
We refer to this unavoidable extra EF as
the \textbf{circuit Landauer loss}.  It equals the sum of the subsystem Landauer losses incurred by each 
non-wire gate's solitary process, 
\begin{align}
\LandauerLoss(p) = \sum_{g \in G \setminus W} \LandauerLoss_g(p)\,,
\label{eq:cll0}
\end{align}
where $\LandauerLoss_g(p) =  S( p_\pag) - S(\mG p_\pag) - \LL_g(p)$.

Each term $ \LandauerLoss_g(p)$ in this sum is non-negative (see end of \cref{sec:solitary}), and
so $\LandauerLoss(p)\ge 0$. 
Note that we can omit wires from the sum in \cref{eq:cll0} because $\mG$ is logically reversible for any wire gate $g$, which means
that $\LandauerLoss_g(p) = 0$ for such gates.

We define \textbf{circuit Landauer cost} 
to be the minimal EF
incurred by running any SR implementation of the circuit, i.e.,
\begin{align}
\LLC(p) &= \LL(p) + \LandauerLoss(p) \label{eq:cllA} \\
& = \sum_{g\in G \setminus W} \Big[ \SSS(\ppag) - \SSS(\mG \ppag) \Big]  \,.
\label{eq:circuit_Landauer_loss}
\end{align}
Recall that $\LL(p)$ is the minimal EF that must be generated by {any} physical process that carries out the map $\map$ on initial distribution $p$.  So by \cref{eq:cllA}, $\LandauerLoss(p)$ is the minimal additional EF that must be generated if we
use an SR circuit to carry
out $\map$ on $p$, no matter how efficient the gates in the circuit are.
In this sense, \cref{eq:cllA} can be viewed as an extension of the generalized Landauer bound, to concern SR circuits.

\noindent \textbf{{3)}}
The third term in \cref{eq:decomp9}, $\MM$, reflects the EF incurred because the actual initial distribution of each gate $g$
is not the optimal one for that gate (i.e., not one that minimizes subsystem EP within each island of the conditional
distribution $P_g$, defined in \cref{eq:spmap}).   
We refer to this cost as the \textbf{circuit mismatch cost}, and write it as
\begin{align}
\label{eq:circmm}
\!\!\MM(p) \!=\! \sum_{\mathclap{g  \in G}} 
\big[\DDb{\pSG}{\qSG} \!-\! \DDb{P_g \pSG}{P_g \qSG} \big] 
\end{align}
where the \textbf{prior}
$\qSG$ is a distribution over $\sX_\SG$ whose conditional distributions over
the islands $\l \in \LLL(P_g)$ all
obey $\EPg(\qSG^c) = \min_{{\rr : \supp \rr \subseteq c}} \EPg(\rr)$.  
Note that we must include wire gates $g$ in the sum in \cref{eq:circmm} even though $\pi_g$ for a wire gate
is logically reversible. This is because 
the associated overall map over $\SG$, \cref{eq:spmap}, is not logically reversible over $\SG$ 
\footnote{There are several
ways that we could manufacture wires that would allow us to exclude them from the sum in \cref{eq:circmm}. One way would be to 
modify the conditional distribution $P_g$ of the wire gates, replacing the logically irreversible \cref{eq:spmap} with the logically reversible
\unexpanded{$P_g(x_{\SG}'\vert x_{\SG}) =
\delta(x_g', x_{\pag}) \delta(x_{\pag}', x_g)$} (thus, each wire gate would basically ``flip'' its input and output). 
Since in an SR circuit an initial state $x_g \ne \nullstate$ should not arise for any gate $g$, 
such modified wire gates would always end up performing the same logical operation as 
would wire gates that obey \cref{eq:spmap}.   
Another way that wire gates
could be excluded from the sum in \cref{eq:circmm} is if their priors had the form $\qSG(x_\pag, x_g) = p_\pag(x_\pag)\delta(x_g, \nullstate)$ (see {\apprefpartialsupport} for details).
}.

$\MM$ is non-negative, since each gate's subsystem mismatch cost is non-negative.
Moreover, $\MM$ approaches its minimum value of 0 as $\pSG^{c_g} \to \qSG^{c_g}$ for all $g\in G$ and all islands ${c_g} \in \LLL(P_g)$.  
(Recall that subsystem priors like $\qSG^{c_g}$ reflect the specific details of the underlying physical process
that implements the gate $g$, such as how its energy spectrum evolves as it runs.) %

Suppose that one wishes to construct a physical system
to implement some circuit, and can vary the associated subsystem priors $\qSG^{c_g}$ arbitrarily.
Then in order to minimize mismatch cost one should choose priors 
 $\qSG^{c_g}$ that equal the actual associated initial distributions $\pSG^c$. Moreover,
those actual initial distributions $\pSG^{c_g}$ can be calculated from the circuit's
wiring diagram, together with the input distribution of the entire circuit,
$\pIN$, by ``propagating'' $\pIN$ through the transformations specified by the wiring
diagram. As a result, given knowledge of the wiring diagram and the input distribution of the entire
circuit, in principle the priors can be set so that mismatch cost is arbitrarily small. %

\vspace{5pt}

\noindent \textbf{{4)}}
The fourth term in \cref{eq:decomp9}, $\RR$, reflects the remaining EF incurred by running the SR circuit, and
so we call it \textbf{circuit residual EP}. Concretely, it equals the subsystem EP that would be incurred even if the initial distribution 
within each island of each gate were optimal:
\begin{align}
\label{eq:circRR}
\RR(p)=\sum_{g\in G}  \sum_{{\l \in \LLL(P_g)}} \pSG(c) \, \subEP_g(q^c_{\SG}) .
\end{align}

Circuit residual EP is non-negative, since each $\subEP_g$ is non-negative. %
Since for every gate $g$, $\pSG(c)$ is a linear function of the initial distribution 
to the circuit as a whole, circuit residual EP also depends linearly on the initial distribution.
Like the priors of the gates, the residual EP terms $\{ \subEP_g(q^c_{\SG})\}$
reflect the ``nitty-gritty'' details of how the gates run.

To summarize, the EF incurred by a circuit can be decomposed into the Landauer cost (the contribution to the EF that would
arise even in a thermodynamically reversible process)
plus the EP (the contribution to that EF which is thermodynamically irreversible). In turn,
there are three contributions to that EP: 
\begin{enumerate}
\item Circuit Landauer loss, which is independent of how the circuit is physically implemented, but does depend
on the wiring diagram of the circuit, the conditional distributions implemented by the gates, and the
initial distribution over inputs.
 It is a nonlinear function of the distribution over inputs, $\pIN$.
\item Circuit mismatch cost, which \textit{does} depend on how the circuit is physical implemented (via the priors),   
as well as the wiring diagram. It is also a nonlinear function of $\pIN$.
\item Circuit residual EP,  which also depends on how the circuit is physical implemented.
It is a {linear} function of $\pIN$.
However, no matter what the wiring diagram of the circuit is,
if we implement each of the gates in a circuit with a quasistatic process, then
the associated circuit residual EP is identically zero, independent of $\pIN$~\footnote{To see this,
simply consider \cref{eq:22a} in \cref{ex:1}, and recall that it is possible to choose a time-varying rate matrix to implement any desired
map over any space in such a way that the resultant EF equals the drop in entropies~\cite{owen_number_2018}.}.
\end{enumerate}

There are other useful decompositions of the EP incurred by an SR circuit that incorporate the wiring diagram.
One such alternative decomposition, which is our second main result, 
leaves the circuit Landauer loss term in \cref{eq:decomp9} unchanged, but modifies the circuit mismatch cost and the circuit residual EP terms.
\begin{theorem}
\label{thm:altdecomp}
The total EF incurred by running an SR circuit where $p$ is the initial distribution over the joint state of all nodes in the circuit  is

\begin{align}
\Q(p) = \LL(p) + \underbrace{\LandauerLoss(p) + \MM'(p)+ \RR'(p)}_{\EP(p)}.
\label{eq:altdecomp2}
\end{align}
\end{theorem}

To present this decomposition, 
recall from \cref{eq:spmap} that for any gate $g$, the distribution over $\sX_{\SG}$ has partial support
at the beginning of the solitary process that implements $P_g$, since there is 0 probability that $x_g \ne \nullstate$. 
%
%
We use this fact 
to apply \cref{thm:cost} to \cref{eq:def_sub_EP}, while taking $\sZ = \{  x_\SG \in \sX_\SG : x_g = \nullstate \}$.  This allows us 
to express the modified circuit mismatch cost by replacing all of the map $P_g(x'_{\SG} \vert x_{\SG})$
in the summand in \cref{eq:circmm} with $P_g(x'_{g} \vert  x_{\pag}) = \pi_g$:
\begin{align}
\MM'(p) &= \sum_{\mathclap{g  \in G \setminus W}} \big[D(\ppag \Vert q_\pag) - D(\mG \ppag \Vert \mG q_\pag)\big]
\label{eq:circuit_mismatch_alt}
\end{align}
where the priors $q_\pag$ are defined in terms of the island decompositions of the 
associated conditional distributions $\mG$, rather than in terms of the island decompositions of the conditional distributions 
$P_g$. 
Note that can exclude the wire gates from the sum in \cref{eq:circuit_mismatch_alt} because each wire gate's $\mG$ is logically reversible, 
and so the associated drop in KL divergence are zero. Then, the modified circuit residual EP is  
\begin{align}
\RR'(p)=\sum_{g\in G}  \sum_{{\l \in \LLL(\mG)}} \ppag \, \subEP_g(q^c_{\pag}) ,
\label{eq:circuit_residual_EP_alt}
\end{align}
where each term $ \subEP_g(q^c_{\pag})$ is given by appropriately modifying the arguments in \cref{eq:def_sub_EP}. In deriving \cref{eq:circuit_residual_EP_alt}, we used the fact that $L_\sZ(P_g) = L(\mG)$.

As with the analogous results in the previous section, \cref{thm:maindecomp} and \cref{thm:altdecomp} differ, because they define ``optimal initial distribution'' relative to different sets of possibilities.  In particular, the decomposition in \cref{thm:maindecomp}  will generally have a larger mismatch cost and smaller residual EP term than the decomposition in \cref{thm:altdecomp}.   

For the rest of this section, 
we will use the term ``circuit mismatch cost'' to refer to the expression in \cref{eq:circuit_mismatch_alt}
rather than the expression in  \cref{eq:circmm}, and similarly will use the term ``circuit residual EP'' to refer
to the expression in \cref{eq:circuit_residual_EP_alt} rather than the expression in \cref{eq:circRR}.

\subsection{Information theory and circuit Landauer loss}
\label{sec:alt_decomp_one_pass_circ}

By combining
\cref{eq:cll0} and \cref{eq:drop_in_mutual}, we can write circuit Landauer loss as
\begin{align}
\label{eq:mmmi}
\LandauerLoss(p) = \sum_{\mathclap{g \in G\setminus W}} \left[I(X_{\pag} ; X_{V \setminus \pag}) - I(X_{g} ; X_{V \setminus g}) \right]
\end{align}
Any nodes that belong to $V \setminus \SG$
 and that are in their initialized state when gate $g$ starts to run will not contribute to the
drop in mutual information terms in \cref{eq:mmmi}. Keeping track of such nodes and simplifying establishes 
the following:
\begin{corollary}
\label{corr:3}
The circuit Landauer loss is
\begin{align*}
\LandauerLoss(p) = \II(\pIN) - \II(\mC \pIN) - \sum_{g \in G \setminus W} \II(\ppag) .
\end{align*}
\end{corollary}
\noindent (We remind the reader that $\II(p_A)$ refers to the multi-information between the variables indexed by $A$.) 


%
%
%
%
%
%
%
%

%

%
%
%
%
%
%
%
%
%
%
%
%
%

%
%
%

\cref{corr:3} suggest a set of novel optimization problems for 
how to design SR circuits: given some desired computation $\mC$ and some given initial distribution $\pIN$, find the circuit wiring diagram that carries out $\mC$ while minimizing the circuit Landauer loss.  
Presuming we have a fixed input distribution $p$ and map $\mC$, 
the term $\II(\pIN) - \II(\mC \pIN)$ in \cref{corr:3} is an additive constant that doesn't depend on the particular choice of circuit wiring diagram. So the optimization problem can be reduced to 
finding which wiring diagram 
results in a minimal value of $ \sum_{g \in G \setminus W} \II(\ppag)$.
In other words, for fixed $p$ and map $\mC$, to minimize the Landauer loss
we should choose the wiring diagram for which the parents of each gate are as strongly correlated among themselves
as possible. Intuitively, this
ensures that the ``loss of correlations'', as information
is propagated down the circuit, is as small as possible.


In general, the distributions over the outputs of the gates in any
particular layer of the circuit will affect the distribution of the inputs of all of the downstream gates,
in the subsequent layers of the circuit. This means that the sum of multi-informations 
in \cref{corr:3} is an inherently global property of the wiring diagram of a circuit; it cannot
be reduced to a sum of properties of each gate considered in isolation, independently of the other gates. 
This makes the optimization problem particularly challenging. We illustrate this optimization problem 
in the following example.

%
%

\begin{figure}
  \includegraphics[trim=0 50 0 0,clip,width=1\linewidth]{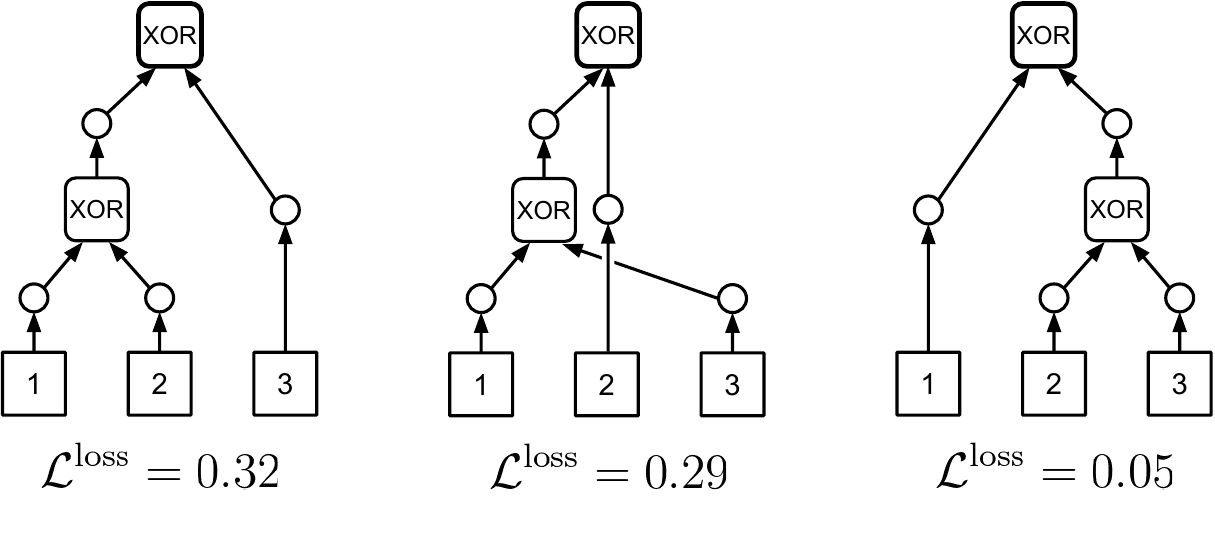}
  $\LandauerLoss = 0.32$
  \qquad\qquad\;
  $\LandauerLoss = 0.29$
  \qquad\qquad\;
  $\LandauerLoss = 0.05$

  \caption{Three different circuit wiring diagrams for computing three-bit parity using XOR gates. The circuit Landauer loss for each wiring diagram is shown below, given the input distribution of \cref{eq:xorspin}.}
  \label{fig:archsearch}
\end{figure}

\begin{example}
\label{ex:paritysearch}
Consider again the case where we want our circuit to compute the three-bit parity function using $2$-input XOR gates, i.e.,
it want it to implement the map
\begin{multline*}
\mC(x_g \vert x_1, x_2, x_3) = \\
\delta(x_g, \delta(x_1 + x_2 + x_3, 1) + \delta(x_1 + x_2 + x_3, 3))
.\end{multline*}
Suppose we happen to know that the input distribution to the circuit (which is specified by how we will use the circuit) is
\begin{align}
\pIN(x_1,x_2,x_3) = \frac{1}{Z} e^{-[\phi({x}_1) \phi({x}_2)/4 + \phi({x}_1)\phi({x}_3) / 2 + \phi({x}_2)\phi({x}_3)]},
\label{eq:xorspin}
\end{align}
where $Z$ is a normalization constant and $\phi(x) := 2x - 1$ is a function that maps bit values $x\in\{0,1\}$ to spin values, $\phi(x) \in \{-1,1\}$.  The distribution \cref{eq:xorspin} is a Boltzmann distribution for a pairwise spin model, where inputs 2 and 3 have the strongest coupling strength ($1$), inputs 1 and 3 have intermediate-strength coupling strength ($1/2$), and inputs 1 and 2 have the weakest coupling strength ($1/4$).

We wish to find the wiring diagram connecting our XOR gates that has minimal circuit Landauer cost for this distribution
over its three input bits.
It turns out that we can restrict our search to three possible wiring diagrams, which are shown in \cref{fig:archsearch}. 
We indicate the circuit Landauer loss for the input distribution of \cref{eq:xorspin} for each of those three wiring
diagrams.  
So for this input distribution, the right-most wiring diagram results in minimal circuit Landauer loss.  Note that this wiring diagram 
aligns with the correlational structure of the input distribution (given that inputs 2 and 3 have the strongest statistical correlation).
\end{example}

An interesting variant of the optimization problem described above arises if we model the residual EP terms for the wire gates.
In any SR circuit, wire gates carry out a logically reversible operation on their inputs.  Thus, by \cref{eq:epdecomp1}, 
\textit{all} of the EF generated by any
wire gates is residual EP. 
If we allow the physical lengths of wires to vary, then as a simple model we could presume that the residual EP
of any wire is proportional to its length. This would allow us to incorporate into our analysis
the thermodynamic effect of
the geometry with which a circuit is laid out on a two-dimensional circuit boards, in
addition to the thermodynamic effect of the topology of that circuit.

Finally, note that for any set of nodes $A$, multi-information can be bounded as $0 \le \II(p_A) \le \sum_{v \in A} S(p_v)
\le \sum_{v \in A} |\sX_v|$. Given this, \cref{corr:3} implies
\begin{align}
\LandauerLoss(p) \le \II(\pIN) \le \ln |\sX_\Cin| \,.
\label{eq:iibound}
\end{align}
This means that for a fixed input state space, the circuit Landauer loss
cannot grow without bound as we vary the wiring diagram. 
Interestingly, this bound on Landauer loss only holds for SR circuits that
have out-degree 1.  If we consider SR circuits that have out-degree greater than 1,
then the circuit Landauer cost can be arbitrarily
large. This is formalized as the following proposition (See proof in \apprefoutdegreebig.)

\begin{restatable}{proposition}{propcircuitscansuck}
For any $\mC$, non-delta function input distribution $\pIN$, 
and $\kappa \ge 0$, there exist an SR circuit with out-degree greater than 1 
that implements $\mC$
for which $\LandauerLoss(p) \ge \kappa$.
\label{prop:circuits_can_suck}
\end{restatable}

\subsection{Information theory and circuit mismatch loss}
\label{sec:circuit_mismatch_loss}

Landauer loss captures the gain in minimal EF due to using an SR circuit, if there is no mismatch cost or residual EP. It is harder to make general statements about the gain in \textit{actual} EF due to using an SR circuit, i.e., when the mismatch cost is nonzero. In this subsection we make some preliminary remarks about this issue.

Imagine that we wish to build a physical process that implements  some computation $\mC(x_\Cout\vert x_\Cin)$ over a space $\cal{X}_\Cin \times \cal{X}_\Cout$.  Suppose we want this process to achieve minimal EP when run with inputs generated by $\ppIN$ (e.g., if  we expect future inputs to the process to be generated by sampling $\ppIN$), and as usual assume the initial value of  $x_\Cout$ will be $\nullstate$ whenever it is run.  Using the decomposition of \cref{eq:systemwide2} and assuming that the residual EP of the process is zero, the EF that such a process would generate if it is actually run with an input distribution $p$ (initialized like SR circuits are, so that it has the form of \cref{eq:p_IN})
is given by the sum of the Landauer cost and the mismatch cost, with no Landauer loss term. 
We write this as
\begin{align}
\Qao(p) = \LL(p) + \DDf{\pIN}{\ppIN} - \DDf{\mC \pIN}{\mC \ppIN},
\label{eq:efmm0}
\end{align}  

Note that in order for the EF generated by an actual physical process to 
be given by \cref{eq:efmm0}, the prior of that process must be $\ppIN$, and in general this may require that the process couple together arbitrary sets of variables. This means that the EF generated by implementing $\mC(x_\Cout\vert x_\Cin)$ with an SR circuit cannot obey \cref{eq:efmm0} in general, due to restrictions on what variables can be coupled in such a circuit.  (One can  verify, for example, that the  prior distribution $\ppIN$ of  a circuit consisting of two disconnected bit erasing gates must be a product distribution over the two input bits.)
To emphasize this distinction, we will refer to a process whose EF is given by  \cref{eq:efmm0}) as an ``all-at-once'' (AO) process (indicated by the subscript ``AO'' in \cref{eq:efmm0}).

For practical reasons, it may be quite difficult to construct an AO process that implements $\mC$, and we must use a circuit implementation instead. In particular, even though the circuit as a whole cannot have prior $\ppIN$, suppose we can set the  priors $\qpag$  at its gates by propagating $\ppIN$ through the wiring diagram of the circuit. 
Assuming again that there is zero EP, the EF that must be incurred by any such SR circuit implementation of $\mC$ on input distribution $p$, assuming some particular wiring topology and gate priors, is  given by the decomposition of \cref{thm:altdecomp},
\begin{align}
\QCC(p) = \LL(p) + \LandauerLoss(p) + \MM'(p) .
\label{eq:efmm1}
\end{align}


We now ask: how much larger is this EF incurred by the SR circuit implementation, compared to that of the original AO process?  Subtracting \cref{eq:efmm0} from \cref{eq:efmm1} gives  
\begin{align}
\Delta \Q & =  \QCC(\pIN) - \Qao(\pIN) \nonumber \\ 
&= \LandauerLoss(\pIN) + \DMcirc(\pIN\Vert \ppIN) 
\label{eq:51}
\end{align}
where we have defined $\DMcirc$ as the the difference between the circuit mismatch cost, $ \MM'(p)$, and the mismatch
cost of the AO process. 
We refer to that difference in mismatch costs as  the \textbf{circuit mismatch loss},
and use \cref{eq:circuit_mismatch_alt} to express it as
\begin{align}
& \DMcirc(\pIN \Vert \ppIN)=  \label{eq:multi_divs} \\
& \;\;\IIDDf{\mC \pIN}{\mC \ppIN} -\IIDDf{\pIN}{ \ppIN} + \sum_{\mathclap{g \in G \setminus W}} \IIDDf{\ppag}{ \ppPAG} , \nonumber 
\end{align}
where $\IIDD$ refers to the multi-divergence, defined in \cref{eq:15a}. 
\cref{eq:multi_divs} can be 
compared to  \cref{corr:3}, which expresses the circuit {Landauer} cost
rather than circuit mismatch cost, and involves multi-informations rather than
multi-divergences.

Interestingly, while circuit Landauer loss is non-negative, circuit mismatch loss
can either be positive or negative. In fact, depending on the wiring diagram, $\pIN$ and $\ppIN$, the
sum of circuit mismatch loss and circuit Landauer loss can be negative. 
This means that when the actual input distribution $\pIN$ is different from the prior distribution of the AO process, the ``closest equivalent circuit'' to the AO process may actually incur less EF than the corresponding AO process.    
This occurs because an SR circuit cannot implement some of the prior distributions that an AO process can implement, so the two implementations end up having different priors. 
This is illustrated in the following example.

\begin{example}
\label{ex:4}
Assume the desired computation is the erasure of two bits, $\mC(x_3, x_4 \vert x_1, x_2) = \delta(x_3,0) \delta(x_4, 0)$, where $x_1$ and $x_2$ refers to the input bits, and $x_3$ and $x_4$ refers to the output bits.   
The prior distribution implemented by an AO process is given by $q_\Cin(0, 0) = q(1, 1) = \epsilon < 1/2$, 
and $q_\Cin(0, 1) = q(1, 0) = 1/2 - \epsilon$. The actual input distribution is given by a delta function distribution, $p_\Cin(x_1,x_2)= \delta(x_1,0)\delta(x_2,0)$. 

We now implement this computation using an SR circuit which consists of two disconnected erasure gates. The closest equivalent SR circuit has gate priors given by the uniform marginal
distributions, $q(x_1)=1/2$ and $q(x_2)=1/2$. 
Then
the difference between the EF of the AO process and the SR circuit is
\begin{align*}
\Delta \Q = &  - 	\big[\SSS(\pIN) \!+\! \DDf{\pIN}{\ppIN} \big]\\
& \qquad + \Big[\sum_g \SSS(\ppag) + \DDf{\ppag}{\qpag}\Big] \\
=&  \sum_{\mathclap{x_1, x_2}} \pIN(x_1, x_2) \ln \ppIN(x_1, x_2)   \nonumber \\
&\qquad - \sum_{x_1} p(x_1) \ln q(x_1)- \sum_{x_2} p(x_2) \ln q(x_2) \\ 
=& \ln \epsilon + 2 \ln 2
\end{align*}
This can be made arbitrarily negative by taking $\epsilon$ sufficiently close to zero. Thus, the
EF of the AO process may be arbitrarily larger than the EF of the closest equivalent SR circuit.
\end{example}

\section{Related work}
\label{sec:earlierwork}

The issue of how the thermodynamic costs of a circuit depend on
the constraints inherent in the topology of the circuit has not previously been addressed 
using the tools of %
modern nonequilibrium statistical physics. %
Indeed, this 
precise issue has received very little attention in \emph{any} of the statistical physics literature.
A notable exception was a 1996 paper by Gernshenfeld~\cite{gershenfeld1996signal}, which pointed out 
that all of the thermodynamic analyses of conventional (irreversible) computing architectures at the time
were concerned with properties of individual gates, rather than entire circuits.  
That paper works through some elementary
examples of the thermodynamics of circuits, and analyzes how the
global structure of circuits (i.e., their wiring diagram) affects their thermodynamic properties.  
Gernshenfeld concludes that the ``next step will be to extend the analysis from these
simple examples to more complex systems'' {\footnote{This prescient article even contains a cursory
		discussion of the thermodynamic consequences of providing a sequence of non-IID rather than IID
		inputs to a computational device,
		a topic that has recently received renewed scrutiny~\cite{mandal2012work,Boyd:2018aa,boyd2016identifying,Strasberg2017}.}}.

There are also several papers that do not address circuits, but focus on 
tangentially related topics,
using modern nonequilibrium statistical physics. 
Ito and Sagawa~\cite{ito2013information,ito_information_2015}  considered the thermodynamics of 
(time-extended) Bayesian networks~\cite{koller2009probabilistic,neapolitan2004learning}. They
divided the variables in the Bayes net into two sets: the sequence of states of a particular system through time, which they write as
$X$, and all external variables that interact with the system as it evolves,
which they write as $\cal{C}$. They then derive and investigate an integral fluctuation theorem~\cite{van2015ensemble,seifert2012stochastic,crooks1998nonequilibrium,rao_esposito_my_book_2019}
that relates the EP generated by $X$ and the EP flowing between $X$ and  $\cal{C}$.
(See also~\cite{ito2016information}).)

Note that \cite{ito2013information} focuses on the EP generated by a proper subset of the
nodes in the entire network. In contrast, our results below concern the EP generated
by all nodes. In addition, while \cite{ito2013information} concentrates on an integral
fluctuation theorem involving EP, we give an exact expression for (expected) EP.

Otsubo and Sagawa~\cite{Otsubo2018} considered the thermodynamics of stochastic Boolean
network models of gene regulatory networks. They focused in particular on 
characterizing the information-theoretic and dissipative properties of 3-node motifs.
While their study does concern dynamics over networks, it has little in
common with the analysis in the current paper, in particular due to its restriction to 3-node
systems.

Solitary processes are similar to ``feedback control'' processes, which have attracted much attention in the thermodynamics of information literature~\cite{sagawa2008second,sagawa2012fluctuation,parrondo2015thermodynamics}.  In feedback control processes, there is a subsystem $A$ that evolves while coupled to another subsystem $B$, which is held fixed.  (This joint evolution is often used to represent either $A$ making a measurement of the state of $B$, or the state of $B$ being used to determine which control protocol to apply to $B$.)  It has been shown for feedback control processes that the total EP incurred by the joint $A\times B$ system is the ``subsystem EP'' of $A$, plus the drop in the mutual information between $A$ and $B$~\cite{sagawa2008second}. Formally, this is identical to \cref{eq:drop_in_mutual}.  

Crucially however, in feedback control processes there is no assumption that $A$ and $B$ are physically decoupled. 
(Formally, \cref{eq:fracsystem2} is not assumed.) Therefore the change in mutual information can either be negative or positive in those
processes (the latter occurs, for instance, when $A$ performs a measurement of the state of $B$). In addition, the ``subsystem EP'' in these processes can be negative. For this reason, in feedback control processes there is no simple relationship between subsystem EP and the total EP incurred by the joint $A\times B$ system.  In contrast, in solitary processes $A$ and $B$ are physically decoupled (cf.  \cref{eq:fracsystem2,eq:fracsystem3}). For this reason, in solitary processes subsystem EP is non-negative, as is the drop in mutual information, \cref{eq:drop_in_mutual}, and so each of them is a lower bound on the total EP incurred by the joint $A\times B$ system.

Boyd et al.~\cite{Boyd:2018aa} considered the thermodynamics of ``modular'' systems, which in our terminology are a special type of solitary processes,
with extra constraints imposed. In particular, to derive their
results, \cite{Boyd:2018aa} assumes there is exactly one thermodynamic reservoir (in their case, a heat bath).
That restricts the applicability of their results. 
Nonetheless,  
individual gates in a circuit are
run by solitary processes, and one could require that they in fact be run by modular systems, in order to analyze the thermodynamics
of (appropriately constrained) circuits. 
However, instead of focusing on this issues,~\cite{Boyd:2018aa} focuses on the thermodynamics of
``information ratchets''~\cite{mandal2012work}, 
modeling them as a sequence of iterations of a single solitary process, 
successively processing the symbols on a semi-infinite tape. 
In contrast, we extend the analysis of single solitary processes operating in isolation to analyze full circuits that comprise
multiple interacting 
solitary processes.  

Riechers~\cite{riechers_thermo_comp_book_2018} also contains work related
to solitary processes, assuming a single heat bath,
like~\cite{Boyd:2018aa}.
~\cite{riechers_thermo_comp_book_2018}
exploits the  
decomposition of EP into ``mismatch cost'' plus ``residual EP'' 
introduced in~\cite{kolchinsky2016dependence}, in order to analyze thermodynamic
attributes of a special kind of circuit.
The analysis in that paper %
is not as complicated as either the analysis in the current paper
or the analysis in~\cite{Boyd:2018aa}. That is because 
~\cite{riechers_thermo_comp_book_2018} does not focus on how the thermodynamic costs of running a system
are affected if we impose a constraint on how the system is allowed to operate (e.g., if we require that it use solitary processes).
In addition, the system considered in that paper %
is a very
special kind of circuit: a set of $N$ disconnected gates, working
in parallel, with the outputs of those gates never combined. 
~\cite{wolpert_thermo_comp_review_2019} is a survey article relating many papers in
the thermodynamics of computation. To clarify some of those relationships, it introduces a type of process
related to solitary processes, called ``subsystem processes''.
(See also~\cite{wolpert2018exact}.) 
For the purposes of the current paper though, we need to
understand the thermodynamics specifically of solitary processes.
In addition, being a summary paper,~\cite{wolpert_thermo_comp_review_2019} presents some results
from the arXiv version of the current paper,~\cite{wolpert2018exact}. Specifically,~\cite{wolpert2018exact,wolpert_thermo_comp_review_2019} 
summarize some of the thermodynamics of straight-line circuits subject to the extra restriction (not made in the current paper)
that there only be a single output node.

There is a fairly extensive literature on ``logically reversible circuits'' and their thermodynamic properties~\cite{fredkin1982conservative,drechsler2012reversible,perumalla2013introduction,frank2005introduction,wolpert_thermo_comp_review_2019}. This work is based on the early analysis in \cite{landauer1961irreversibility}, and so it is not grounded
in modern nonequilibrium statistical physics. Indeed, 
modern nonequilibrium statistical physics reveals some important subtleties and caveats with
the thermodynamic properties of logically reversible circuits~\cite{wolpert_thermo_comp_review_2019}. 
Also see~\cite{sagawa2014thermodynamic} for important
clarifications of the relationship between thermodynamic and logical reversibility, not appreciated in
some of the research community working on logically reversible circuits.

Finally, another related paper is \cite{grochow_wolpert_sigact_2018}. This paper
starts by taking a distilled version of the decomposition of EP in~\cite{kolchinsky2016dependence} 
as given. It then discusses some of the many new problems in computer science theory that this decomposition leads to, both
involving circuits and involving many other kinds of computational system.

\section{Discussion and Future Work}
\label{sec:future_work}

It is important to emphasize that SR circuits are somewhat unrealistic models of many real-world digital circuits. For example, many real digital circuits
have multiple gates running at the same time, and often do not reinitialize
their gates after they're run. In addition, many real digital
circuits have characteristics like loops and branching. This makes them challenging to model at all using simple solitary processes. Extending our analysis to these more general models of circuits is an important direction for future work. Nonetheless, it is worth mentioning that 
 all of the thermodynamic  costs discussed above --- including
Landauer loss, mismatch cost, and residual EP ---  are intrinsic to \textit{any} physical process,
as described in~\cref{sec:efdecomp}. So versions of them arise in those other kinds of circuits, only in modified form.

An interesting set of issues to investigate in future work is the scaling properties of the thermodynamic
costs of SR circuits. In conventional circuit complexity theory ~\cite{arora2009computational,savage1998models}
one first specifies a ``circuit family'' which comprises an infinite set of circuits that have different
size input spaces but that are all (by definition) viewed as ``performing the same computation''. For
example, one circuit family is given by an infinite set of circuits
each of which has a different size input space, and
outputs the single bit of whether the number of 1's in its input string is odd or even. Circuit complexity theory
is concerned with how various resource costs in making a given circuit (e.g., the number of
gates in the circuit) scales with the size of the circuit as one goes through the members of a
circuit family. For example, it may analyze how the number of gates in a set of 
of circuits, each of which determines whether its input string contains an odd number of 1's, scales
with the size of those input strings. One interesting set of issues for future research is
to perform these kinds of scaling analyses when the ``resource costs'' are thermodynamic
costs of running the circuit, rather than conventional costs like the number of gates. In particular, it's interesting
to consider classes of circuit families defined in terms of such costs, in analogy to the
complexity classes considered in computer science theory, like $\cs{P/poly}$, or $\cs{P / log}$.

Other interesting issues arise if we formulate a cellular automaton (CA) as a circuit with an infinite number
of nodes in each layer, and an infinite number of layers, each layer of the circuit corresponding to another timestep
of the CA. For example, suppose we are given a particular CA rule (i.e., a particular map taking the state of
each layer $i$ to the state of layer $i+1$) and a particular distribution over its initial infinite bit
pattern. These uniquely specify the ``thermodynamic EP rate'', given by the total EP generated by running the CA for $n$ iterations
(i.e., for reaching the $n$'th layer in the circuit), divided by $n$. It would be interesting
to see how this EP rate depends on the CA rule and initial distribution over bit patterns.

Finally, another important direction for future work arises if we broaden our scope beyond digital circuits designed by
human engineers, to include naturally occurring circuits such as brains and gene regulatory networks. The ``gates'' in such circuits
are quite noisy --- but all of our results hold independent of the noise levels of the gates. On the other hand,
like real digital circuits, these naturally occurring circuits have loops, branching, concurrency, etc., and
so might best be modeled with some extension of the models introduced in this paper. Again though, the
important point is that whatever model is used, the EP generated by running a physical system governed
by that model would include Landauer loss, mismatch cost, and residual EP.

$ $

\begin{acknowledgments}
\emph{Acknowledgments} ---
We would like to thank Josh Grochow for helpful discussion,
and thank the Santa Fe Institute for helping to support this research. This paper was made possible through Grant No. CHE-1648973 from the U.S. National Science Foundation, Grant No. FQXi-RFP-1622 from the Foundational Questions Institute, and Grant No. FQXi-RFP-IPW-1912 from the Foundational Questions Institute and Fetzer Franklin Fund, a donor advised fund of Silicon Valley Community Foundation
\end{acknowledgments}

\bibliographystyle{ieeetr}

\clearpage
%

\appendix

\section{Proof of \cref{thm:cost} and related results}
\label{app:costappendix}

\subsubsection*{Preliminaries}

\newcommand\ei[1][1=i]{{\mathbf{u}^{#1}}}%

\global\long\def\i{x}%
\global\long\def\j{y}%
\global\long\def\Tji{\map(y\vert x)}%

\global\long\def\p{p}%
\global\long\def\s{\mu}%
\global\long\def\N{\mathbb{N}}%
\global\long\def\q{q}%
\global\long\def\N{\mathbb{N}}%
\global\long\def\sX{\mathcal{X}}%
\global\long\def\sY{\mathcal{Y}}%
\global\long\def\hlf{{\textstyle \frac{1}{2}}}%
\global\long\def\e{\epsilon}%
\global\long\def\qip{[\M\q]_{j}}%
\global\long\def\pip{[\M\p](\j)}%
\global\long\def\f{f}%
\global\long\def\optp{s}%
\global\long\def\E{\mathbb{E}}%
\global\long\def\w{q^{\e}}%
\global\long\def\wi{\w_{i}}%
\global\long\def\wip{\left[\M\w\right]_{j}}%
\global\long\def\wp{\M\w}%
\global\long\def\de{{\textstyle \partial_{\e}^{+}}}%
\global\long\def\dem{{\textstyle \partial_{\e}^{-}}}%
\global\long\def\ppi{p(\i)}%
\global\long\def\ppl{p(\l)}%
\global\long\def\ppll{p^{c}}%
\global\long\def\ppil{\ppll(\i)}%
\global\long\def\qe{q^{\e}}%
\global\long\def\pip{p(\j)}%
\global\long\def\qip{q(\j)}%
\global\long\def\qll{q^{c}}%

Consider a conditional distribution $\map(\j\vert\i)$ that specifies
the probability of ``output'' $\j\in\sY$ given ``input'' $\i\in\sX$,
where $\sX$ and $\sY$ are finite. 

Given some $\sZ\subseteq \sX$, the island decomposition $L_\sZ(P)$ of $\map$, and any $p\in\Delta_\sX$,
let $\ppl=\sum_{\i\in\l}\ppi$ indicate the total probability within
island $\l$, and 
\[
\ppil:=\begin{cases}
\frac{\ppi}{\ppl} & \text{if \ensuremath{\i\in\l} and \ensuremath{\ppl>0}}\\
0 & \text{otherwise}
\end{cases}
\]
indicate the conditional probability of state $\i$ within island
$\l$.

In our proofs below, we will make use of the notion of \emph{relative
interior.} Given a linear space $V$, the relative interior of a subset $A\subseteq V$ is defined as \citep{borwein2003notions}
\[
\mathrm{relint}\, A := \{x\in A : \forall y \in A,\exists\epsilon>0 \text{ s.t. } x+\e(x-y)\in A\}\,.
\]

Finally, for any function $g(x)$, we use the notation
\[\partial_{x}^{+}g(x)\vert_{x=a}:=\lim_{\delta \rightarrow0^{+}}\frac{1}{\delta}\left(g(a+\delta)-g(a)\right)
\]
to indicate the right-handed derivative of $g(x)$ at $x=a$. When the condition that $x=a$ is omitted, $a$ is implicitly
assumed to equal $0$, i.e.,
\[\partial_{x}^{+}g(x) :=\lim_{\delta \rightarrow0^{+}}\frac{1}{\delta}\left(g(\delta)-g(0)\right)
\]
\global\long\def\aej{[\map\ae]_{j}}%
\global\long\def\aei{\ae_{i}}%
\global\long\def\aej{[\map\ae](\j)}%
\global\long\def\aei{\ae(\i)}%
\global\long\def\aej{\ae(\j)}%
\global\long\def\aei{\ae(\i)}%
\global\long\def\ae{a^{\e}}%
We also adopt the shorthand that
$\ae\coloneqq a+\e(b-a)$, and write $S(\ae) : = S(p(\ae))$,  $P\ae := Pp(\ae)$, and
so  $S(P\ae) = S(Pp(\ae))$.

\subsubsection*{Proofs}

Given some conditional distribution $\Tji$ and function $\f:\sX\to\mathbb{R}$, we consider the function $\FF:\Delta_\sX \rightarrow\mathbb{R}$
as
\[
\FF(\p):=S(\map \p)-S(\p)+\E_{p}[f]\,.
\]
Note that $\FF$ is continuous on the relative interior of $\Delta_\sX $.

\vspace{5pt}

\begin{applemma}
\label{lem:dd}For any $a,b\in\Delta_\sX$, the directional derivative of
$\FF$ at $a$ toward $b$ is given by
\[
\de\FF(a+\e(b-a))\vert_{\e=0}=D(\map b\Vert\map a)-D(b\Vert a)+\FF(b)-\FF(a).
\]
\end{applemma}
\begin{proof}
Using the definition of $\FF$,  write
\begin{align}
\de\FF(\ae)=\de\left[S(\map\ae)-S(\ae)\right]+\de\E_{\ae}[f].\label{eq:ddd-1}
\end{align}
Consider the first term on the RHS,
\global\long\def\bj{b(\j)}%
\global\long\def\bi{b(\i)}%
\global\long\def\aj{a(\j)}%
\global\long\def\ai{a(\i)}%
\begin{align*}
& \de\left[S(\map\ae)-S(\ae)\right] \\
& =-\sum_{\j\in\sY}\left[ (\de P\aej)\ln P\aej+\de[\map\ae](\j)\right] \\
& \qquad+\sum_{\i\in\sX}\left[(\de \ae)i\ln\aei+\de\aei\right]\\
 & =-\sum_{\j\in\sY}(P\bj\!-\!P\aj)\ln\aej+\sum_{\i\in\sX}(\bi\!-\!\ai)\ln\aei
\end{align*}
Evaluated at $\epsilon=0$, the last line can be written as 
\begin{align*}
 & -\sum_{\j\in\sY}(\bj\!-\!\aj)\ln\aj+\sum_{\i\in\sX}(\bi-\ai)\ln\ai\\
 & \quad=D(\map b\Vert\map a)+S(\map b)-S(\map a)-D(b\Vert a)-S(b)+S(a)
\end{align*}
where we adopt the convention that if $a(x) = 0, b(x) \ne 0$ for some $x$, then this expression means $-\infty$.
We next consider the $\de\E_{\ae}[f]$ term,
\begin{align*}
\de\E_{\ae}[f] &=\de\left[\sum_{\i\in\sX}\left(a(\i)+\e(b(\i)-a(\i))\right)f(\i)\right]\\
&=\E_{b}[f]-\E_{a}[f]\,.
\end{align*}

Combining the above gives
\begin{align*}
\de\FF(\ae)\vert_{\e=0} & =D(\map b\Vert\map a)-D(b\Vert a)+S(\map b)-S(b)\\
& \qquad -(S(\map a)-S(a))+\E_{b}[f]-\E_{a}[f]\\
 & =D(\map b\Vert\map a)-D(b\Vert a)+\FF(b)-\FF(a).
\end{align*}
\end{proof}

Importantly, \cref{lem:dd} holds even if there are $x$ for which $a(x) = 0$ but $b(x) \ne 0$, 
in which case the RHS of the equation in the lemma equals $-\infty$. (Similar comments apply
to the results below.)

\begin{apptheorem}
\label{thm:1}Let $V$ be a convex subset of $\Delta$. Then for any
$q\in\argmin_{\optp\in V}\FF(\optp)$ and any $p\in V$,
\begin{equation}
\FF(p)-\FF(\q)\ge D(p\Vert\q)-D(\map p\Vert\map\q)\,.\label{eq:thm1es}
\end{equation}
Equality holds if $\q$ is in the relative interior of $V$.
\end{apptheorem}
\begin{proof}
Define the convex mixture $\qe:=q+\e(p-q)$. 
By \cref{lem:dd}, the
directional derivative of $\FF$ at $\q$ in the direction $p-\q$
is
\[
\de\FF(\qe)\vert_{\e=0}=D(\map p\Vert\map\q)-D(p\Vert\q)+\FF(p)-\FF(q)\,.
\]
At the same time, $\de\FF(\qe)\vert_{\e=0} \ge 0$, since $q$ is a minimizer within a convex set. 
\cref{eq:thm1es} then follows by rearranging.

When $q$ is in the relative interior of $V$, $q-\epsilon(p-q)\in V$
for sufficiently small $\epsilon>0$. Then,
\begin{align*}
0\le & \lim_{\e\rightarrow0^{+}}\frac{1}{\e}\left(\FF(q-\epsilon(p-q))-\FF(q)\right)\\
= & -\lim_{\e\rightarrow0^{-}}\frac{1}{\e}\left(\FF(q+\epsilon(p-q))-\FF(q)\right)\\
= & -\lim_{\e\rightarrow0^{+}}\frac{1}{\e}\left(\FF(q+\epsilon(p-q))-\FF(q)\right)\\
&=-\de\FF(\qe)\vert_{\e=0}.
\end{align*}
where in the first inequality comes from the fact that $q$ is a minimizer,
in the second line we change variables as $\e\mapsto-\e$,
and the last line we use the continuity of $\FF$ on interior of the simplex.
Combining with the above implies 
\[\de\FF(\w)=D(\map p\Vert\map\q)-D(p\Vert\q)+\FF(p)-\FF(q)=0.\]
\end{proof}

The following result is key. It means that the prior within an island has full support in that island.

\begin{applemma}
\label{lem:fsupport}
For any $\l\in\Lmap$ and $\displaystyle q\in\argmin_{\optp : \supp \optp \subseteq \l }\FF(\optp)$,
\[
\supp q=\{ x \in \l : f(x) < \infty \}.
\]
\end{applemma}
\begin{proof}
\global\long\def\qi{q(\i)}%
\global\long\def\qj{q(\j)}%
\global\long\def\ii{\hat{\i}}%
\global\long\def\jj{\hat{\j}}%
\global\long\def\wj{\w(\j)}%
\global\long\def\wi{\w(\i)}%
\global\long\def\ffi{f(\i)}%
\global\long\def\ffii{f(\ii)}%
\global\long\def\Tjjii{\map(\jj\vert\ii)}%
\global\long\def\wjj{\w(\jj)}%
\global\long\def\wii{\w(\ii)}%
\global\long\def\Tjii{\map(\j\vert\ii)}%
\global\long\def\qjj{q(\jj)}%

We prove the claim by contradiction. Assume that $q$ is a minimizer with $\supp q\subset \{ x \in \l : f(x) < \infty \}$. Note there
cannot be any $\i\in\supp q$ and $\j\in\sY\setminus\supp\map q$
such that $\Tji>0$ (if there were such an $\i,\j$, then $\qj=\sum_{\i'}P(\j\vert\i')q(\i')\ge\Tji\qi>0$,
contradicting the statement that $\j\in\sY\setminus\supp\map q$).
Thus, by definition of islands, there must be an $\ii\in\l\setminus\supp q$,
$\jj\in\supp\map q$ such that $f(\ii) < \infty$ and $\Tjjii>0$.

Define the delta-function distribution $u(\i):=\delta(\i,\ii)$ and
the convex mixture $\qe(\i)=(1-\e)q(\i)+\e u(\i)$
for $\e\in[0,1]$. We will also use the notation $\qe(\j)=\sum_{\i}\Tji\qi$.

Since $q$ is a minimizer of $\FF$, $\partial_{\e}\FF(\qe)\vert_{\e=0}\ge0$.
Since $\FF$ is convex, the second derivative $\partial_{\e}^{2}\FF(\qe)\ge0$
and therefore $\partial_{\e}\FF(\qe)\ge0$ for all $\epsilon\ge0$.
Taking $a=\qe$ and $b=u$ in \cref{lem:dd} and rearranging, we then have
\begin{align}
\FF(u) & \ge D(u\Vert\qe)-D(\map u\Vert\map\qe)+\FF(\qe) \nonumber \\
 & \ge D(u\Vert\qe)-D(\map u\Vert\map\qe)+\FF(q),\label{eq:ineqZ} 
\end{align}
where the second inequality uses that $q$ is a minimizer of $\FF$. At the same time, 
\begin{align}
& D(u\Vert\qe)-D(\map u\Vert\map\qe) \nonumber \\
& =\sum_{\j}P(\j\vert\ii)\ln\frac{\qe(\j)}{\qe(\ii)\map(\j\vert\ii)}\nonumber \\
 & =\Tjjii\ln\frac{\qe(\jj)}{\e\Tjjii}+\sum_{\j\ne \jj}P(\j\vert\ii)\ln\frac{\qe(\j)}{\e\map(\j\vert\ii)} \nonumber \\
 & \ge\Tjjii\ln\frac{(1-\e)q(\jj)}{\e\Tjjii}+\sum_{\j\ne \jj}P(\j\vert\ii)\ln\frac{\e\map(\j\vert\ii)}{\e\map(\j\vert\ii)} \nonumber \\
 & =\Tjjii\ln\frac{(1-\e)}{\e}\frac{q(\jj)}{\Tjjii} \label{eq:ineqZZZ} ,
\end{align}
where in the second line we've used that $\qe(\ii)=\e$, and in the third that $\qe(\j)=(1-\e)q(\j) + \e \map(\j\vert \ii)$, so $\qe(\j)\ge(1-\e)q(\j)$ and $\qe(\j) \ge \e \map(\j\vert \ii)$.

Note that the RHS of \cref{eq:ineqZZZ} goes to $\infty$ as $\e \to 0$.  Combined with \cref{eq:ineqZ} and that $\FF(q)$ is finite implies that $\FF(u) = \infty$.  However, $\FF(u) = S(\map(Y\vert \ii)) + f(\ii) \le |\sY| +f(\ii)$, which is finite.  We thus have a contradiction, so $q$ cannot be the minimizer.
\end{proof}

The following result is also key. Intuitively,   it follows from the fact that the directional derivative of $S(p)$ into the
simplex for any $p$ on the edge of the simplex is negative infinite.

\begin{applemma}
\label{lem:unique}
For any island $\l\in \Lmap$, $\displaystyle q\in\argmin_{s : \supp s \subseteq \l}\FF(p)$
is unique.
\end{applemma}
\begin{proof}
Consider any two distributions $\displaystyle p,q\in\argmin_{s:  \supp s \subseteq \l}\FF(s)$, and let $p' = \map p$, $q'=\map q$. We will prove that $p=q$.

First, note that by \cref{lem:fsupport}, $\supp q = \supp p = c$.  By \cref{thm:1},
\begin{align*}
\FF(p)-\FF(q)&=D(p \Vert q)-D(p' \Vert q') \\
& = \sum_{\i,\j} p(x) \Tji \ln \frac{p(\i) q'(\j)}{q(\i)p'(\j)} \\
& = \sum_{\i,\j} p(x) \Tji \ln \frac{p(\i) \Tji}{q(\i)p'(\j) \Tji / q'(\j)} \\
& \ge 0
\end{align*}
where the last line uses the log-sum inequality.  If the inequality is strict, then  $p$ and $q$ can't 
both be minimizers, i.e., the minimizer must be unique, as claimed. 

If instead the inequality is not strict, i.e., $\FF(p)-\FF(q) = 0$, then there is some constant $\alpha$ such that for all $\i,\j$ with $\Tji > 0$,
\begin{align}
 \frac{p(\i) \Tji}{q(\i)p'(\j) \Tji / q'(\j)}=\alpha  
\end{align}
which is the same as
\begin{align}
 \frac{p(\i)}{q(\i)} = \alpha \frac{p'(\j)}{q'(\j)}.
 \label{eq:lsecond}
\end{align}

Now consider any two different states $x,x'\in c$ such that
$\map(y\vert x)>0$ and $\map(y\vert x')>0$ for some $y$ (such states
must exist by the definition of islands). For \cref{eq:lsecond} to hold for both $x, x'$ with that same, shared $y$, it must be that 
${p(x)}/{q(x)}={p(x')}/{q(x')}$. Take another state
$x''\in c$ such that $\map(y'\vert x'')>0$ and $\map(y'\vert x')>0$
for some $y'$. Since this must be true for all pairs $x, x' \in c$, 
${p(x)}/{q(x)}=\text{const}$ for all $x\in\l$, and $p=q$, as claimed.
\end{proof}

\begin{applemma}
\label{thm:decomp}$\FF(p)=\sum_{\l\in\Lmap}\ppl\FF(\ppll)$.
\end{applemma}
\begin{proof}
First, for any island $\l\in\Lmap$, define 
\[
\phi(\l)=\{\j\in\sY:\exists\i\in\l\text{ s.t. }\Tji>0\}\,.
\]
In words, $\phi(\l)$ is the subset of output states in $\sY$ that
receive probability from input states in $\l$. By the definition
of the island decomposition, for any $\j\in\phi(\l)$, $\Tji>0$ only
if $\j\in\l$. Thus, for any $p$ and any $\j\in\phi(\l)$, we can
write
\begin{equation}
\frac{\pip}{\ppl}=\frac{\sum_{\i}\Tji\ppi}{\ppl}=\sum_{\i\in\sX}\Tji\ppil\,.\label{eq:id0}
\end{equation}

Using $p=\sum_{\l\in\Lmap}\ppl\ppll$
and linearity of expectation, write $\E_{p}[f]=\sum_{\l\in\Lmap}\ppl\E_{\ppll}[f]$.
Then,
\begin{align*}
& S(\map p)-S(p) \\
& =-\sum_{\j}\pip\ln\pip+\sum_{\i}\ppi\ln\ppi\\
 & =\sum_{\l\in\Lmap}\ppl\Big[-\sum_{\j\in\phi(\l)}\frac{\pip}{\ppl}\ln\frac{\pip}{\ppl}+\sum_{\i\in\l}\frac{\ppi}{\ppl}\ln\frac{\ppi}{\ppl}\Big]\nonumber \\
 & =\sum_{\l\in\Lmap}\ppl\left[S(\map\ppll)-S(\ppll)\right],
\end{align*}
where in the last line we've used \cref{eq:id0}. Combining gives
\begin{align*}
\FF(p)&=\sum_{\l\in\Lmap}\ppl\left[S(\map\ppll)-S(\ppll)+\E_{\ppll}[f]\right]\\
& =\sum_{\l\in\Lmap}\ppl\FF(\ppll)\,.
\end{align*}
\end{proof}

We are now ready to prove the main result of this appendix.

\thmmaincost*
\begin{proof}
We prove the theorem by considering two cases separately.
\vspace{5pt}

\noindent \textbf{Case 1}: $\sZ = \sX$. This case can be assumed when $f(x)<\infty$ for all $x$, so that $L_\sZ(\map) = L(\map)$.  
Then, by \cref{thm:decomp}, we have $\FF(p)=\sum_{\l\in\Lmap}\ppl\FF(\ppll)$.
By \cref{lem:fsupport} and \cref{thm:1},
\[
\FF(\ppll)-\FF(\qll)=D(\ppll\Vert\qll)-D(\map\ppll\Vert\map\qll),
\]
where we've used that if some $\supp\qll=\l$, then $\qll$ is in
the relative interior of the set $\{ s \in \Delta_\sX : \supp s \subseteq \l \}$. $\qll$ is unique by \cref{lem:unique}.

At the same time, observe
that for any $p,r\in\Delta_\sX$,
\global\long\def\ri{r(\i)}%
\global\long\def\rj{r(y)}%
\begin{align*}
& D(p\Vert r)-D(\map p\Vert\map r)\\
& = \sum_{\i}\ppi\ln\frac{\ppi}{\ri}-\sum_{\j}\pip\ln\frac{\pip}{\rj}\\
&= \sum_{\l\in\Lmap}\ppl\Bigg[\sum_{\i\in\l}\frac{\ppi}{\ppl}\ln\frac{\ppi/\ppl}{\ri/r(\l)} \\
& \qquad\qquad\qquad\qquad-\sum_{\j\in\phi(\l)}\frac{\pip}{\ppl}\ln\frac{\pip/\ppl}{\rj/r(\l)}\Bigg]\\
 &= \sum_{\l\in\Lmap}\ppl\left[D(\ppll\Vert r^{\l})-D(\map\ppll\Vert\map r^{\l})\right]\,.
\end{align*}
The theorem follows by combining.

\vspace{5pt}

\noindent \textbf{Case 2}: $\sZ \subset \sX$. In this case, define
 a ``restriction'' of $f$ and $\map$ 
to domain $\sZ$ as follows:
\begin{enumerate}
\item Define $\tilde{f} : \sZ \to \mathbb{R}$ via $\tilde{f}(\i)=f(x)$ for $\i\in \sZ$. 
\item Define the conditional distribution $\tilde{\map}(\j\vert\i)$ for $\j\in \sY,\i \in \sZ$ via $\tilde{\map}(\j\vert\i)=\Tji$ for all $y\in\sY,\i\in \sZ$. 
\end{enumerate}
In addition, for any distribution $p \in \Delta_\sX$ with $\supp p \subseteq \sZ$, let $\tilde{p}$ be a distribution over $\sZ$ defined via $\tilde{p}(\i)=p(\i)$ for $\i \in \sZ$.  Now, by inspection, it can be verified that for any  $p \in \Delta_\sX$ with $\supp p \subseteq \sZ$,
\begin{align}
\FF(p) = S(\tilde{\map}\tilde{p})-S(\tilde{p})+\E_{\tilde{p}}[\tilde{f}] =: \tilde{\FF}(\tilde{p})
\label{eq:www99}
\end{align}
We can now apply \text{Case 1} of the theorem to the function $\tilde{\FF}:\Delta_\sZ\to\mathbb{R}$, as  defined in terms of the tuple $({\sZ}, \tilde{f}, \tilde{\map})$ (rather than the function $\FF:\Delta_\sX\to\mathbb{R}$, as defined in terms of the tuple 
$(\sX, {f}, {\map})$).  This gives 
\begin{align}
\tilde{\FF}(\tilde{p}) = D(\tilde{p}\Vert \tilde{q}) -  D(\tilde{\map}\tilde{p}\Vert \tilde{\map} \tilde{q}) + \sum_{\l \in \LLL(\tilde{\map})} \tilde{p}(\l) \tilde{\FF}(\tilde{q}^\l),
\label{eq:B1}
\end{align}
where, for all $\l \in \LLL(\tilde{\map})$, $\tilde{q}^\l$ is the unique distribution that satisfies
$\tilde{q}^\l \in \argmin_{r \in \Delta_\sZ :\supp r \subseteq \l} \tilde{\FF}(r)$.

Now, let $q$ be the natural extension of $\tilde{q}$ from $\Delta_\sZ$ to $\Delta_\sX$.  Clearly,  for all $\l \in \LLL(\tilde{\map})$, ${\FF}({q}^\l)=\tilde{\FF}(\tilde{q}^\l)$ by \cref{eq:www99}. In addition, each $q^c$  is the unique distribution that satisfies 
${q}^\l \in \argmin_{r \in \Delta_\sX :\supp r \subseteq \l} {\FF}(r)$. 
Finally, it is easy to verify that $D(\tilde{p}\Vert \tilde{q}) = D({p}\Vert {q})$, $D(\tilde{\map}\tilde{p}\Vert \tilde{\map} \tilde{q}) = D({\map}{p}\Vert {\map} {q})$, $L(\tilde{\map}) = L_\sZ(\map)$ (recall the definition of $L_\sZ$ from \cref{sec:island_def}).  Combining the above results with \cref{eq:www99} gives
 \[
 \FF(p) = \tilde{\FF}(\tilde{p}) = D({p}\Vert {q}) -  D({\map}{p}\Vert {\map} {q}) + \sum_{\mathclap{\l \in \LLL_\sZ({\map})}} {p}(\l) {\FF}({q}^\l).
 \]
\end{proof}

\begin{example}
Suppose we are interested in thermodynamic costs associated with 
functions $f$ whose image contains the value infinity, i.e., $\f:\sX\to\mathbb{R}\cup\{\infty\}$.
For such functions,  $\FF(p)=\infty$ for any $p$ which has support over an $x\in \sX$ such that $\f(x)=\infty$.  
In such a case it is not meaningful to consider a prior distribution $q$ (as in \cref{thm:cost}) which has support over any $x$ with $\f(x)=\infty$. 
For such functions we also are no longer able to presume that the optimal distribution has full support within each island of $\l \in \Lmap$, because in general the proof of \cref{lem:fsupport} no longer holds when $f$ can take infinite values.

Nonetheless, by \cref{eq:B1}, for the
purposes of analyzing the thermodynamic costs of actual initial distributions $p$ that have finite $\FF(p)$ (and so have zero mass on any $x$ such
that $f(x) = \infty$), we can always carry out our usual analysis if we first reduce the problem to an appropriate ``restriction'' of $f$.
\end{example}

\begin{example}
\label{ex:ctmchiddencost}
	Suppose we wish to implement a (discrete-time) dynamics $\map(x'\vert x)$ over $\sX$ using a CTMC.
	Recall from the end of \cref{sec:stoch_thermo} that by appropriately expanding the state space $\sX$ to include a set of
	``hidden states'' $\sZ$ in addition to $\sX$, and appropriately designing the rate matrices over that expanded state space $\sX \cup \sZ$,
	we can ensure that the resultant evolution over $\sX$ is arbitrarily close to the
	desired conditional distribution $\map$. Indeed, one can even design those rates matrices over  $\sX \cup \sZ$
	so that not only is the dynamics over $\sX$ arbitrarily close
	to the desired $\map$, but in addition the EF generated in running that CTMC  over  $\sX \cup \sZ$ is arbitrarily close 
	to the lower bound of \cref{eq:fracsystem3}~\cite{owen_number_2018}.
	
	However, in any real-world system that implements some $P$ with a CTMC over an expanded space $\sX \cup \sZ$,
	that lower bound will not be achieved, and nonzero EP will be generated. In general, to analyze the EP of such 
	real-world systems one has
	to consider the mismatch cost and residual EP of the full CTMC over the expanded space $\sX \cup \sZ$.
	Fortunately though, we can design the CTMC over $\sX \cup \sZ$ so that when it begins the implementation of $\map$, 
	there is zero probability mass on any of the states in $\sZ$~\cite{wolpert_spacetime_2019,owen_number_2018}. 
	If we do that, then we can apply  \cref{eq:B1}, and so
	restrict our calculations of mismatch cost and residual EP to only involve
	the dynamics over $\sX$, without any concern for the dynamics over $\sZ$. 
\end{example}

\begin{example}
Our last example is to derive the alternative decomposition of the EP of an SR circuit which is discussed in \cref{sec:alt_decomp_EP}.
Recall that due to	\cref{eq:spmap}, the initial distribution over any gate in an SR circuit
has partial support. This means we can apply \cref{eq:B1} to decompose the EF,
in direct analogy to the use of \cref{thm:cost} to derive	\cref{thm:maindecomp} --- only 
with the modification that the spaces $X$ and $Y$ are set to
$\sX_{\pag}$ and $\sX_g$, respectively, rather than both set to $\sX_{\SG}$, as was done in
deriving \cref{thm:maindecomp}. (Note that the islands 
also change when we apply  \cref{eq:B1} rather than \cref{thm:cost}, from the islands of $P_g$ to the islands of $\pi_g$). 
The end result is a decomposition of EF just like that
in \cref{thm:maindecomp}, in which we have the same circuit Landauer cost and  circuit Landauer loss expressions as in that theorem,
but now have the modified forms of circuit mismatch cost
and of circuit residual EP introduced in  \cref{sec:alt_decomp_EP}.
\label{ex:alt_decomp_SR_circuit}
\end{example}

\section{Thermodynamics costs for SR circuits}
\label{app:singlepass}

To begin, we will make use of the fact that there is no overlap in time among the solitary processes in an SR circuit, so 
the total EF incurred can be written as 
\begin{align}
\Q(p) = \sum_{g \in G} \Q_g(\pSG) .
\label{eq:c1}
\end{align}
Moreover, for each gate $g$, the solitary process that updates the variables in $\SG$ 
starts with $\variablevalue_g$ in its initialized state with probability $1$. So we can 
overload notation and write $ \Q_g(\ppag)$ instead of 
$\Q_g(\pSG)$ for each gate $g$.

\subsubsection*{Derivation of \cref{thm:maindecomp,eq:circuit_Landauer_loss}}

Apply \cref{thm:cost} to \cref{eq:def_sub_EP}
to give
\begin{multline}
\label{eq:appdecomp3}
\EPg(\pSG) = \DDb{\pSG}{\qSG} - \DDb{P_g \pSG}{P_g \qSG} \\
+ \sum_{\mathclap{\l \in \LLL(P_g)}} \pSG(c) \EPg(\qSG^c) \,,
\end{multline}
where $\qSG$ is a distribution that satisfies 
\[\qSG^c \in  \argmin_{{\rr : \supp \rr \subseteq c}} \EPg(\rr)
\]  for all islands $\l \in \LLL(P_g)$. 
Next, for convenience use \cref{eq:c1} to write $\Q(p)$ as
\begin{align}
\Q(p) = \LL(p) +  \Big(\bigg[\sum_{g \in G} \Q_g(\ppag)\bigg] - \LL(p)\Big) \label{eq:app0}
\end{align}
The basic decomposition of $\Q(p)$ given in \cref{thm:maindecomp} into
a sum of $\LL$ (defined in \cref{eq:spl0}), $\LandauerLoss$ (defined in \cref{eq:cll0}), $\MM$ (defined in \cref{eq:circmm}), and $\RR$ (defined in \cref{eq:circRR}) 
comes from combining \cref{eq:app0,eq:def_sub_EP,eq:appdecomp3} and then grouping 
and redefining terms.


Next, again use the fact that the solitary processes have no
overlap in time to establish that the minimal value of the sum of the EPs of the gates
is the sum of the minimal EPs of the gates considered separately of one another.
As a result, we can jointly take $\EPg(\pSG) 
\rightarrow 0$ for all gates $g$ in the circuit~\cite{owen_number_2018}.
We can then use \cref{eq:c1} to establish that
the minimal EF of the circuit is simply the sum of the minimal
EFs of running each of the gates in the circuit, i.e., the sum of
the subsystem Landauer costs of running the gates.
In other words, the circuit Landauer cost is
\eq{
\LLC(p) &=\sum_{g \in G} \Big[ \SSS(\pSG) - \SSS(P_g \pSG) \Big]   \\
&= \sum_{g\in G} \Big[ \SSS(\ppag) - \SSS(\mG \ppag) \Big]   \\
&= \sum_{g\in G \setminus W} \Big[ \SSS(\ppag) - \SSS(\mG \ppag) \Big]   \,.
}
To derive the second line, we've used the fact that in an SR
circuit, each gate is set to its initialized value at the beginning of its solitary process
with probability $1$, and that its parents are set to {their} initialized states
with probability $1$ at the end of the process.
Then to derive the third line we've used the fact that wire gates implement
the identity map, and so $\SSS(\ppag) - \SSS(\mG \ppag) = 0$ for
all $g \in W$. 

This establishes \cref{eq:circuit_Landauer_loss}.

\subsubsection*{Derivation of \cref{eq:mmmi}}

To derive \cref{eq:mmmi}, we first write Landauer cost as
\begin{align*}
 \LL(p) = \SSS(p) - \SSS(P p) = \sum_{g \in G} [\SSS(p^\begG_V) - \SSS(p^\ennG_V) ] \,. 
\end{align*}
We  then write circuit Landauer loss as
\begin{align}
& \LandauerLoss(p) \nonumber \\
&= \sum_{g \in G} \Big[ \SSS(\pSG) - \SSS(P_g \pSG) \Big] - \LL(p) \nonumber \\
& = \sum_{g \in G} \Big[ (\SSS(\pSG) - S(p_V^\begG)) - (\SSS(P_g \pSG)- S(p_V^\ennG)) \Big] \nonumber \\
& = \sum_{g \in G} \Big[ (\SSS(\pSG) + S(p_{V\setminus \SG}^\begG) - S(p_V^\begG)) - \nonumber \\
& \qquad\qquad (\SSS(P_g \SG) + S(p_{V\setminus \SG}^\ennG) - S(p_V^\ennG)) \Big] \label{eq:appmi11} \\
& = \sum_{g \in G}  \left[I_{p^\begG}(X_{\SG} ; X_{ V\setminus \SG}) - I_{p^\ennG}(X_{\SG} ; X_{ V\setminus \SG}) \right] ,\label{eq:appmi15}
\end{align}
In \cref{eq:appmi11}, we used the fact that a solitary process over $\SG$ leaves the nodes in $V\setminus \SG$ unmodified, thus $S(p_{V\setminus \SG}^\begG)=S(p_{V\setminus \SG}^\ennG)$. 

Given the assumption that $\variablevalue_g = \nullstate$ at the beginning of the solitary process for gate $g$, we can rewrite 
\begin{align}
I_{p^\begG}(X_{\SG} ; X_{ V\setminus \SG}) = I_{p^\begG}(X_{\pag} ; X_{ V\setminus \pag})\,. \label{eq:appmi16}
\end{align}
Similarly, because $X_v = \nullstate$ for all $v \in \pag$ at the end of the solitary process for gate $g$, we can rewrite
\begin{align}
I_{p^\ennG}(X_{\SG} ; X_{ V\setminus \SG}) = I_{p^\ennG}(X_{g} ; X_{ V\setminus g})\,. \label{eq:appmi17}
\end{align}
Finally, for any wire gate $g\in W$, given the assumption that $X_g = \nullstate$ at the beginning of the solitary process, we can write
\begin{align}
I_{p^\begG}(X_{\pag} ; X_{ V\setminus \pag}) =  I_{p^\ennG}(X_{g} ; X_{ V\setminus g})\,. \label{eq:appmi18}
\end{align}

\cref{eq:mmmi} then follows from combining \cref{eq:appmi15,eq:appmi16,eq:appmi17,eq:appmi18} and simplifying.


\subsubsection*{Derivation of \cref{corr:3,eq:iibound}}
First, write circuit Landauer loss as
\begin{align}
\LandauerLoss(p) = \sum_{\mathclap{g\in G \setminus W}} \Big[ \SSS(\ppag) - \SSS(p_g) \Big] - \LL(p) .
\label{eq:cllapp1}
\end{align}
Then, rewrite the sum in \cref{eq:cllapp1} as 
\begin{align}
& \sum_{\mathclap{g\in G \setminus W}} \Big[ \SSS(\ppag) - \SSS(p_g) \Big] \nonumber \\
& = \sum_{\mathclap{g\in G \setminus W}}  \SSS(\ppag) - \sum_{\mathclap{g\in G \setminus W}}  \SSS(p_g)  \nonumber \\
& = \sum_{\mathclap{g\in G \setminus W}}  \SSS(\ppag) - \sum_{\mathclap{g\in G \setminus (W\cup \Cout)}}  \SSS(p_g) -\sum_{\mathclap{g \in \Cout}} \SSS(p_g)\nonumber  \\
& = \sum_{\mathclap{g\in G \setminus W}}  \SSS(\ppag) - \sum_{\mathclap{v\in V \setminus (W\cup \Cout)}}  \SSS(p_v) + \sum_{\mathclap{v \in \Cin}} \SSS(p_v)  -\sum_{\mathclap{g \in \Cout}} \SSS(p_g). \label{eq:cllapp2}
\end{align}
Now, notice that for every $v\in V \setminus (W\cup \Cout)$ (i.e., every node which is not a wire and not an output), there is a corresponding wire $w$ which transmits $v$ to its child, and which has $\SSS(p_w) = \SSS(p_v)$.  This lets us rewrite \cref{eq:cllapp2} as
\begin{align}
&\sum_{\mathclap{g\in G \setminus W}}  \SSS(\ppag) - \sum_{\mathclap{w\in W}}  \SSS(p_w) + \sum_{\mathclap{v \in \Cin}} \SSS(p_v)  -\sum_{\mathclap{g \in \Cout}} \SSS(p_g) \nonumber \\
&=\sum_{\mathclap{g\in G \setminus W}} \Big[ \SSS(\ppag) - \sum_{\mathclap{v \in \pag}} \SSS(p_v) \Big] + \sum_{\mathclap{v \in \Cin}} \SSS(p_v) - \sum_{\mathclap{g \in \Cout}} \SSS(p_g)\nonumber \\
&= - \sum_{g \in G} \II(\ppag) + \sum_{\mathclap{v \in \Cin}} \SSS(p_v) - \sum_{\mathclap{g \in \Cout}} \SSS(p_g) \label{eq:cllapp3}
\end{align}
where in the second line we've used the fact
that every wire belongs to exactly one set $\pag$ for a non-wire gate $g$, and in the last line we used the definition of multi-information. 
Then, using the definition $\LL(p)=\SSS(\pIN) - \SSS(p_\Cout)$, the definition of multi-information, and by combining \cref{eq:cllapp1,eq:cllapp2,eq:cllapp3}, we have
\begin{align}
\LandauerLoss(p)  = \II(\pIN) - \II(\mC \pIN) - \sum_{g \in G} \II(\ppag) .
\label{eq:app9}
\end{align}

To derive \cref{eq:iibound}, note that
\begin{multline}
\II(\pIN) = \Big[\sum_{v \in \Cin} \SSS(p_v)\Big] - \SSS(\pIN) 
\le  \sum_{v \in \Cin} \SSS(p_v) \\ \le \sum_{v \in \Cin} \ln \vert \sX_v \vert = \ln \left\vert \prod_{v \in \Cin}  \sX_v  \right\vert = \ln \vert \sX_\Cin \vert \,.
\end{multline}

\section{SR circuits with out-degree greater than 1}
\label{app:outdegreebig}

In this appendix, we consider a more general version of SR circuits, in which non-output gates can have out-degree greater than $1$.

First, we need to modify the definition of an SR circuit in \cref{sec:singlepass}.
This is because in SR circuits, 
the subsystem corresponding to a given gate $g$ reinitializes all of the parents of that gate to their initialized state, $\nullstate$. If, however, there is some node $v$ that has out-degree greater than $1$ --- i.e., has more than one child --- then 
we must guarantee that no such $v$ is reinitialized by one its children gates before all of its children gates have run.  
To do so, we require that each non-output node $v$ in the circuit
is reinitialized only by the \textit{last} of its children gates 
to run, while the earlier children (if any) apply the identity map to $v$.

Note that this rule could result in different thermodynamic costs of an overall circuit, 
depending on the precise topological order we use to determine which of the children of a given $v$ reinitialized $v$. 
This would mean that the
entropic costs of running a circuit would depend on the (arbitrary) choice
we make for the topological order of the gates in the circuit. 
This issue won't arise in this paper however. To see why, recall that 
we model
the wires in the circuit themselves as gates, which have 
both in-degree and out-degree equal to 1. As a result, if $v$ has out-degree greater than $1$,
then $v$ is not a wire gate, and therefore all of its children must be 
wire gates --- and therefore none of those children has multiple parents. So the problem
is automatically avoided.

We now prove that for SR circuits with out-degree greater than 1, circuit Landauer loss can be arbitrarily large.

\propcircuitscansuck*
\begin{proof}
Let $\CC = (V, E, F,\sX)$ be such a circuit that implements $\mC$. 
Given that $p$ is not a delta function, 
there must be an input node, which we call $v$, such that $\SSS(p_{v})>0$. Take $g \in \Cout$ to be any output gate of $\CC$, and let $\mG \in F$ be its update map.   

Construct a new circuit $\CC' = (V', E', F', \sX')$, as follows:
\begin{enumerate}
\item $V' = V \cup \{w',g',w''\}$;
\item $E' = E \cup \{(v, w'), (w', g'), (g',w''), (w'', g)\}$;
\item $F' = (F \; \setminus \mG) \; \cup \; \{\pi_{w'}, \pi_{w''}, \pi_{g'}, \mG' \}$ where 
\begin{align*}
\pi_{w'}(x_{w'} \vert x_v) & = \delta(x_{w'}, x_v)\\
\pi_{w''}(x_{w''} \vert x_{g'}) & = \delta(x_{w''}, x_{g'})\\
\pi_{g'}(x_{g'} \vert x_{w}) &= \delta(x_{g'}, \nullstate)\\
\mG'(x_g | x_{\pag}, x_{w''}) &= \mG(x_g \vert x_{\pag}).
\end{align*}
\item $\sX_{w'} = \sX_{g'} =\sX_{w''}= \sX_v$.
\end{enumerate}
In words, $\CC'$ is the same as $\CC$ except that: (a) we have added an ``erasure gate'' $g'$ which takes $v$ as input (through a new wire gate $w'$), and (b) this erasure gate is provided as an additional input, which is completely ignored, to one of the existing output gates $g$ (through a new wire gate $w''$).

It is straightforward to see that $\mC'=\mC$. At the same time, $\SSS(p_{\pa(g')}) - \SSS(\Pi_{g'} p_{\pa(g')}) = \SSS(p_v)$, thus 
\begin{align}
\LandauerLoss_{\CC'}(p) = \LandauerLoss_\CC(p) + \SSS(p_{v}) \,,
\end{align}
where $\LandauerLoss_{\CC'}$ and $\LandauerLoss_{\CC}$ indicate the circuit Landauer loss of $\CC'$ and $\CC$ respectively.   
This procedure can be carried out again to create a new circuit $\CC''$ from $\CC'$, which
also implements $\mC$ but which now has Landauer loss $\LandauerLoss_{\CC''}(p) = \LandauerLoss_\CC(p) + 2\SSS(p_{v})$. Iterating, we can construct a circuit with an arbitrarily large Landauer loss which implements $\mC$.
\end{proof}

\clearpage

\end{document}